\documentclass{article}
\pdfoutput=1

\usepackage{arxiv}

\usepackage[utf8]{inputenc} % allow utf-8 input
\usepackage[T1]{fontenc}    % use 8-bit T1 fonts
\usepackage{hyperref}       % hyperlinks
\usepackage{url}            % simple URL typesetting
\usepackage{booktabs}       % professional-quality tables
\usepackage{amsfonts}       % blackboard math symbols
\usepackage{nicefrac}       % compact symbols for 1/2, etc.
\usepackage{microtype}      % microtypography
\usepackage{lipsum}
\usepackage{graphicx}
\usepackage{array}
\usepackage{amssymb,srcltx}%,amsthm
\usepackage[centertags]{amsmath}
\usepackage{color}
\usepackage{times}
\usepackage{enumerate}
\usepackage{bm}
\usepackage{hyperref}
\usepackage{accents}
\usepackage{natbib}
\usepackage{amsthm}

\newcommand{\iid}{\stackrel {{\rm iid.}}{\sim}}
\newcommand{\ind}{\stackrel {{\rm ind.}}{\sim}}

\newcommand{\sumas}{\sum^n_{i=1}}

\newcommand{\ii}{i=1,\ldots,n}

\newtheorem{theorem}{Theorem}
\newcommand{\balpha}{\mbox{\boldmath $\alpha$}}
\newcommand{\bmu}{\mbox{\boldmath $\mu$}}
\newcommand{\bphi}{\mbox{\boldmath $\phi$}}

\newcommand{\bnu}{\mbox{\boldmath $\nu$}}
\newcommand{\bSigma}{\mbox{\boldmath $\Sigma$}}
\newcommand{\bepsilon}{\mbox{\boldmath $\epsilon$}}
\newcommand{\bLambda}{\mbox{\boldmath $\Lambda$}}
\newcommand{\bbeta}{\mbox{\boldmath $\beta$}}
\newcommand{\btheta}{\mbox{\boldmath $\theta$}}

\newcommand{\bDelta}{\mbox{\boldmath $\Delta$}}
\newcommand{\bdelta}{\mbox{\boldmath $\delta$}}
\newcommand{\bzeta}{\mbox{\boldmath $\zeta$}}
\newcommand{\NI}{\textrm{SMN}}
\newcommand{\N}{\textrm{N}}
\newcommand{\SNI}{\textrm{SMSN}}
\newcommand{\SN}{\textrm{SN}}
\newcommand{\ST}{\textrm{ST}}
\newcommand{\SSL}{\textrm{SSL}}

\newcommand{\A}{\mathbf{A}}
\newcommand{\bPsi}{\mbox{\boldmath $\Psi$}}

\newcommand{\bOmega}{\mbox{\boldmath $\Omega$}}

\newcommand{\bupsilon}{\mbox{\boldmath $\upsilon$}}

\newcommand{\bGamma}{\mbox{\boldmath $\Gamma$}}

\newcommand{\be}{\mathbf{b}}

\newcommand{\yp}{\mathbf{y}}

\newcommand{\y}{\mathbf{y}}
\newcommand{\Y}{\mathbf{Y}}
\newcommand{\bD}{\mathbf{D}}

\newcommand{\blambda}{\mbox{\boldmath $\lambda$}}

\newcommand{\Z}{\mathbf{Z}}

\newcommand{\X}{\mathbf{X}}
\newcommand{\bW}{\mathbf{W}}
\newcommand{\bH}{\mathbf{H}}
\newcommand{\E}{\textrm{E}}
\newcommand{\ba}{\mathbf{b}}
\newcommand{\app}{\stackrel {{\rm .}}{\sim}}
\title{Approximate inferences for nonlinear
mixed effects models with scale mixtures of skew--normal distributions}

\author{
    Fernanda L. Schumacher \\
    Department of Statistics \\
    Campinas State University\\
    Campinas, SP - Brazil\\
    \texttt{fernandalschumacher@gmail.com}\\
    \And 
    Dipak K. Dey\\
    Department of Statistics\\ University of Connecticut\\
    Storrs, CT - U.S.A.\\
    \texttt{dipak.dey@uconn.edu}\\
    \And 
    Victor H. Lachos\\
    Department of Statistics\\ University of Connecticut\\
    Storrs, CT - U.S.A.\\
    \texttt{hlachos@uconn.edu}\\
}

\begin{document}
\maketitle

\begin{abstract}
Nonlinear mixed effects models have
received a great deal of attention in the statistical literature in
recent years because of their flexibility in handling longitudinal
studies, including human immunodeficiency virus viral dynamics,
pharmacokinetic analyses, and studies of growth and decay. A
standard assumption in nonlinear mixed effects models for continuous
responses is that the random effects and the within-subject errors are normally distributed, making the model sensitive to outliers. We present a novel class of asymmetric nonlinear mixed effects models that provides efficient parameters estimation in the
analysis of longitudinal data. We assume that, marginally, the
random effects follow a multivariate scale mixtures of skew--normal
distribution and that the random errors follow a symmetric
scale mixtures of normal distribution, providing an appealing robust
alternative to the usual normal distribution. We propose an approximate method for maximum likelihood estimation based on an EM-type algorithm that produces approximate maximum likelihood estimates and significantly reduces the numerical difficulties associated with the exact maximum likelihood estimation. Techniques for prediction of
future responses under this class of distributions are also briefly
discussed. The methodology is illustrated through an application
to Theophylline kinetics data and through some simulating studies.
\end{abstract}

% keywords can be removed

\section{Introduction}

This is the birth centenary year of the living legend and giant in the world of
statistics, Prof. C.R. Rao. This article is a partial reflection
of Dr. Rao’s contributions to statistical theory and methodology, including  sufficiency, efficiency of estimation, as well as the application of matrix theory in linear statistical inference and beyond. In this paper, we extend many results from linear  models to nonlinear mixed effects (NLME) models which have been receiving notable
attention in recent statistical literature, mainly due to their flexibility for dealing with longitudinal data and repeated measures data. In a NLME framework it
is routinely assumed that the random effects and the within--subject
measurement errors follow a normal distribution. While this assumption
makes the model easy to apply in widely used software (such as R and SAS), its accuracy is difficult to check and the routine use of normality has been questioned by many authors. For example,
\cite{HartforDav200} showed through simulations that inference
based on the normal distribution can be sensitive to underlying
distributional and model misspecification. \cite{Litiere2007} showed
the impact of misspecifying the random effects distribution on the
estimation and hypothesis testing  in generalized linear mixed
models. Specifically, they showed that the maximum likelihood
estimators are inconsistent in the presence of misspecification and
{that the estimates of the variance components are severely biased}.
More recently, \cite{hui2020random} showed through theory and simulation that under misspecification, standard likelihood ratio tests of truly non‐zero variance components can suffer from severely inflated type I errors, and confidence intervals for the variance components can exhibit considerable under coverage.
Thus it is of practical interest to explore frameworks with considerable flexibility in the
distributional assumptions of the random effects as well as the
error terms, which can produce more reliable inferences.

There has been considerable work in mixed effects models in this
direction. \cite{VerbeLes1996} introduced a heterogeneous linear
mixed model (LMM) where the random effects distribution is relaxed
using normal mixtures. \cite{Pinheiro01} and \cite{lin2017multivariate} proposed a multivariate
Student-$t$ linear and nonlinear (T--LMM/NLMM) mixed model, respectively, and showed that it performs
well in the presence of outliers. \cite{ZhangDavi} proposed a LMM in
which the random effects follow a so--called semi--nonparametric %(SNP) 
distribution. \cite{Rosa2003} adopted a Bayesian framework to
carry out posterior analysis in LMM with the thick--tailed class of
normal/independent distributions. Moreover, \cite{Lachos_Ghosh_Arellano_2009} proposed a
skew--normal independent linear mixed model based on the
scale mixtures of skew--normal (SMSN) family introduced by
\cite{Branco_Dey01}, developing a general EM--type
algorithm for maximum likelihood estimation (MLE).  

In the nonlinear context, \cite{lachos2013bayesian} considered the Bayesian estimation of NLME models with scale mixtures of normal (SMN) distributions for the error term and random effects, 
\cite{lachos2011linear} developed a Bayesian framework for censored linear and nonlinear mixed effects models replacing the Gaussian assumptions for the random terms with SMN  distributions, and \cite{delacruz2014bayesian} also considered a Bayesian framework to estimate NLME models under heavy-tailed distributions, allowing the mixture variables associated with errors
and random effects to be different. From a frequentist perspective, \cite{meza2012estimation} proposed an estimation procedure to obtain the maximum likelihood estimates for NLME models with NI distributions, and \cite{galarza2020quantile} developed a likelihood-based approach for estimating quantile regression models with correlated continuous longitudinal data using the asymmetric Laplace distribution, both using a stochastic approximation of the EM algorithm.
Furthermore, \cite{Cibele09} and \cite{pereira2019nonlinear} considered a NLME model with skewed and heavy-tailed distributions, with the limitation that the nonlinearity is incorporated only in the fixed effects. 

Extending the work of \cite{Lachos_Ghosh_Arellano_2009}, in this
paper we propose a parametric robust modeling of NLME models based
on SMSN distributions. In particular, we assume a mean--zero SMSN
distribution for the random effects, and a SMN distribution for the
within--subject errors. Together, the observed responses follow
conditionally an approximate SMSN distribution and define what we
call a scale mixtures of skew--normal nonlinear mixed effects (SMSN--NLME)
model. In particular, the SMSN distributions provide a class of
skew--thick--tailed distributions that are useful for robust
inference and that contains as proper elements the skew--normal
(SN),  skew--$t$ (ST), skew--slash (SSL), and the skew--contaminated normal (SCN) distributions. The marginal density of the response variable %of the observed quantities 
can be obtained by approximations, leading to a computationally efficient
approximate (marginal) likelihood function that can be implemented
directly by using existing statistical software. 
The hierarchical representation of the proposed model makes the implementation of an efficient EM--type algorithm possible, which results in ``closed form'' expressions for the E and M--steps.
%{We draw attention to the fact that since we are considering a
%mean-zero SNI distribution for the random effects,{\color{red} all the
%procedures for the linearization problem are different than those presented in
%\cite{Lachos_Ghosh_Arellano_2009}.}}%vague?

{The rest of the article is organized as follows. The SMSN--NLME
model is presented in Section \ref{sec2}, including a brief introduction to
the class of SMSN distributions and the approximate likelihood-based methodology for
inference in our proposed model. In Section \ref{sec3} we propose an EM-type
algorithm for approximate likelihood inferences in SMSN--NLME models,
which maintains the simplicity and stability of the EM--type
algorithm proposed by \cite{Lachos_Ghosh_Arellano_2009}.
In Section \ref{sec4}, simulation studies are conducted to evaluate the empirical performance of the proposed model. The advantage of the proposed methodology is
illustrated through the Theophylline kinetics data in Section \ref{sec5}.
Finally, some concluding remarks are presented in Section \ref{sec6}.}

\section{The model and approximate likelihood}\label{sec2}

\subsection{SMSN distributions and main notation}

The idea of the SMSN distributions originated from an early work by
\cite{Branco_Dey01}, which included the skew--normal (SN)
distribution as a special case. We say that a $p\times 1$ random vector
$\textbf{Y}$  follows a SN distribution with $p\times 1$ location
vector $\bmu$, $p\times p$ positive definite dispersion matrix
$\bSigma$ and $p\times 1$ skewness parameter vector $\blambda,$ and
write $\textbf{Y}\sim SN_p(\bmu,\bSigma,\blambda),$ if its
probability density function (pdf) is given by
\begin{eqnarray}\label{denSN}
f(\mathbf{y})&=& 2{\phi_p(\mathbf{y};\bmu,\bSigma)
\Phi(\blambda^{\top}\y_0)},
\end{eqnarray}
where $\y_0=\bSigma^{-1/2}(\mathbf{y}-\bmu)$,
$\phi_p(.;\bmu,\bSigma)$ stands for the pdf of the $p$--variate
normal distribution with mean vector $\bmu$ and dispersion matrix
$\bSigma$, $N_p(\bmu,\bSigma)$ say, and $\Phi(.)$ is the cumulative
distribution function (cdf) of the standard univariate normal.  Letting
$\Z=\Y-\bmu$ and noting that $a\Z\sim SN_p(\textbf{0},a^2\bSigma,\blambda)$ for all scalar $a>0$,
we can define a SMSN distribution as that of a $p-$dimensional random vector
\begin{equation}\label{stoNI} \Y=\bmu+U^{-1/2}\Z,
\end{equation}
where $U$ is a positive random variable with the cdf $H(u;\bnu)$ and
pdf $h(u;\bnu)$, and independent of the
$SN_p(\textbf{0},\bSigma,\blambda)$ random vector $\Z$, with $\bnu$
being a scalar or vector parameter indexing the distribution of the
mixing scale factor  $U$. Given $U=u$, $\Y$ follows a multivariate
skew--normal distribution with location vector $\bmu$, scale matrix
$u^{-1}\bSigma$ and skewness parameter vector $\blambda$. Thus, by
(\ref{denSN}), the marginal pdf of $\Y$ is
\begin{eqnarray}\label{denSNI}
f(\mathbf{y})&=&
2\int^{\infty}_0{\phi_p(\mathbf{y};\bmu,u^{-1}\bSigma)
\Phi(u^{1/2}\blambda^{\top}\y_0)}dH(u;\bnu).
\end{eqnarray}
The notation $\mathbf{Y}\sim \SNI_p(\bmu,\bSigma,\blambda;H)$ will
be used when $\Y$ has pdf (\ref{denSNI}).

The class of SMSN distributions includes the skew--$t$, skew--slash, and skew--contaminated normal, which will be briefly introduced subsequently. All these distributions have heavier tails than the skew-normal and can
be used for robust inferences. When $\blambda=\mathbf{0}$, the SMSN
distributions reduces to the SMN class, i.e.,
the class of scale--mixtures of the normal distribution, which is represented
by the pdf
$f_0(\mathbf{y})=\int^{\infty}_0{\phi_p(\mathbf{y};\bmu,u^{-1}\bSigma)
}dH(u;\bnu)$ and will be denoted by $\textrm{\NI}_p(\bmu,\bSigma,H)$. We refer to \cite{Lachos_Ghosh_Arellano_2009} for
details and additional properties related to this class of
distributions.

\begin{itemize}
\item {\it Multivariate skew--$t$ distribution}

The multivariate skew--$t$ distribution with $\nu$  degrees of
freedom, denoted by $\ST_{{p}}(\bmu,\bSigma,\mathbf{\blambda};\nu)$, can be
derived from the mixture model (\ref{denSNI}), by taking $U\sim$
$Gamma(\nu/2,\nu/2),$ $\nu>0.$ The pdf of $\Y$ is
\begin{equation}
f(\mathbf{y})=2t_p(\mathbf{y};\bmu,\bSigma,\nu)T\left(\sqrt{\frac{\nu+p}{\nu+d}}\A;\nu+p\right),\,\,\,\,\,\mathbf{y}\in
\mathbb{R}^p,\label{lsdefAB1}
\end{equation}
where $t_p(\cdot;\bmu,\bSigma,\nu)$ and $T(\cdot;\nu)$ denote,
respectively, the pdf of the $p$--variate Student--$t$ distribution,
namely $t_p(\bmu,\bSigma,\nu)$, and the cdf of the standard
univariate $t$--distribution,
$\A=\blambda^{\top} \bSigma^{-1/2} (\mathbf{y}-\bmu)$ and
$d=(\y-\bmu)^{\top}\bSigma^{-1}(\y-\bmu)$ is the Mahalanobis
distance.

\item {\it Multivariate skew--slash distribution}

Another SMSN distribution, termed as the multivariate skew--slash
distribution and denoted by $\SSL_p(\bmu,\bSigma,\blambda;\nu)$,
arises when the distribution of $U$ is $Beta(\nu,1)$, $\nu>0$. Its
pdf is given by
\begin{equation}\label{densSSL}
f(\mathbf{y})=2\nu\int^1_0u^{\nu-1}\phi_p(\mathbf{y};\bmu,u^{-1}\bSigma)\Phi(u^{1/2}
\A)du,\,\,\,\,\,\mathbf{y}\in \mathbb{R}^p.
\end{equation}
The skew--slash distribution reduces to the skew--normal
distribution as $\nu\uparrow \infty$.

\item {\it Multivariate skew--contaminated normal distribution}

The multivariate skew--contaminated normal distribution, denoted by
$SCN_p(\bmu,\bSigma,\mathbf{\blambda};\nu_1,\nu_2),$ arises when
the mixing scale factor $U$ is a discrete random variable taking one
of two values. The pdf of $U$, given a parameter vector
$\bnu=(\nu_1,\nu_2)^{\top}$, is
\begin{equation}\label{uSNC} h(u;\bnu)=\nu_1
\mathbb{I}_{(u=\nu_2)}+(1-\nu_1)
\mathbb{I}_{(u=1)},\,\,\,0<\nu_1<1,\,0<\nu_2< 1.\end{equation} It
follows that
\begin{eqnarray}
f(\mathbf{y})&=&2\left\{\nu_1\phi_p(\mathbf{y};\bmu,\nu_2^{-1}{\bSigma})
\Phi(\nu_2^{1/2}\A)\nonumber+(1-\nu_1)
\phi_p(\mathbf{y};\bmu,\bSigma)\Phi(\A)\right\}.
\end{eqnarray}
\end{itemize}

\subsection{The SMSN--NLME model}

In this section, we present the general NLME model proposed in this work, in which the random terms are assumed to follow a SMSN distribution within the class defined in \eqref{stoNI}. The model, denoted by SMSN--NLME, can be defined as follows:
\begin{eqnarray}
\textbf{Y}_i&=&\eta(\bphi_i,\X_i)+\bepsilon_i,\,\,\,\bphi_i=\A_i\bbeta+\ba_i,\label{modeleq11}
\end{eqnarray}
{with the assumption that}
\begin{equation}\label{modSnmis1} \left(\begin{array}{c}
        \ba_i \\
        \bepsilon_i
      \end{array}
\right)\buildrel ind.\over\sim
\SNI_{q+n_i}\left(\left(\begin{array}{c}
                               c\bDelta \\
                               \mathbf{0}
                             \end{array}
\right),\left(\begin{array}{cc}
           \mathbf{D} & \mathbf{0} \\
           \mathbf{0} & \sigma^2_e \mathbf{I}_{n_i}
         \end{array}
\right),\left(\begin{array}{c}
                               \mathbf{\blambda} \\
                               \mathbf{0}
                             \end{array}
\right) ;H \right),\,\,\ii,
\end{equation}
where the subscript $i$ is the subject index,
$\textbf{Y}_i=(y_{i1},\cdots,y_{in_i})^{\top}$ is an $n_i\times 1$
vector of observed continuous responses for subject $i$,
$\eta$ represents a nonlinear vector-valued differentiable function of the individual mixed effects parameters $\bphi_i$, $\X_i$ is an $n_i\times q$ matrix of covariates, $\bbeta$ is
a $p\times 1$ vector of fixed effects, $\be_i$ is a $q$-dimensional
random effects vector associated with the $i$th subject,
$\textbf{A}_i$ is a $q\times p$ design matrix
that possibly depends on elements of $\X_i$, $\bepsilon_i$ is
the $n_i\times 1$ vector of random errors,
$c=c(\bnu)=-\sqrt{\frac{2}{\pi}} k_1$, with $k_{1} = E\{U^{-1/2}\}$, and $\bDelta = \textbf{D}^{1/2} \bdelta$, with $\bdelta={\blambda}/{\sqrt{1+\blambda^{\top}\blambda}}$. The dispersion matrix
$\textbf{D}=\textbf{D}(\balpha)$ depends on unknown and reduced
parameter vector $\balpha$. Finally, as was indicated in the
previous section, $H=H(\cdot|\bnu)$ is the cdf-generator that
determines the specific $\SNI$ model that is considered.

\vspace*{0.1in} \noindent\texttt{\bf Remarks:}

\begin{itemize}
    \item[\bf i)] The model  defined in (\ref{modeleq11}) can be
viewed as a slight modification of the general NLME model proposed by
\cite{PinheiroBates95} and \cite{PinheiroBates2000}, with the
restriction that our new model does not allow to incorporate, for
instance, {\it ``time-varing'' } covariates in the random effects. This assumption is made for simplicity of theoretical derivations. However, the methodology proposed here can be extended
without any difficulty.

\item[\bf ii)] An attractive and convenient way to specify
(\ref{modSnmis1}) is the following:
\begin{equation}\label{condMod}
\textbf{b}_i|U_i=u_i\buildrel ind.\over\sim
\SN_q(c\bDelta,u_i^{-1}\mathbf{D},\blambda),\;\bepsilon_i|U_i=u_i
\buildrel ind.\over\sim \N_{n_i}(\mathbf{0},\sigma^2_e
u_i^{-1}\mathbf{I}_{n_i}),
\end{equation}
which are independent, where $U_i \iid H$.
Since for each $i=1,\ldots,n,$ $\textbf{b}_i$ and $\bepsilon_i$ are
indexed by the same scale mixing factor $U_i,$ they are not
independent in general. Independence corresponds to the case when
$U_i=1$ $(\ii),$ so that the SMSN--NLME model reduces to the SMN--NLME
model as defined in \cite{lachos2013bayesian}. However, conditional on
$U_i,$ $\textbf{b}_i$ and $\bepsilon_i$ are independent for each
$i=1,\ldots,n,$ which implies that $\mathbf{b}_i$ and $\bepsilon_i$
are uncorrelated, since
$Cov(\mathbf{b}_i,\bepsilon_i)=\E\{\mathbf{b}_i\bepsilon^{\top}_i\}
=\E\{\E\{\mathbf{b}_i\bepsilon^{\top}_i|U_i\}\}=\mathbf{0}$. Thus,
it follows from (\ref{modSnmis1})-(\ref{condMod}) that marginally
\begin{equation}
\be_i \iid \SNI_q(c\bDelta,\mathbf{D},\blambda;
H)\quad\textrm{and}\quad\bepsilon_i\buildrel ind.\over\sim
\NI_{n_i}(\mathbf{0},\sigma^2_e \mathbf{I}_{n_i}; H),\quad\ii.
 \label{modeleq2}
\end{equation}
Moreover, as long as $k_{1}<\infty$ the chosen location parameter ensures that $E\{\be_i\}=E\{\bepsilon_i\}=\mathbf{0}$. Thus, this model
considers that the within-subject random errors
are symmetrically distributed, while the distribution of random effects is assumed to be asymmetric and to have mean zero.

\item[\bf iii)] Our model can be seen as an extension of the
elliptical NLME model proposed by \cite{Cibele09}, where the
nonlinearity is incorporated only in the fixed effects. If
$\eta(.)$ is a linear function of the individual mixed effects parameters
$\bphi_i$, then the SMSN--NLME model reduces to a slight modification
of the SNI--LME model proposed by \cite{Lachos_Ghosh_Arellano_2009}.
However, since in this work we consider a mean-zero SMSN distribution for the random effects,
the result given in \cite{Lachos_Ghosh_Arellano_2009} cannot be directly applied.
One the other hand, this choice of location parameter is important, since $E\{\be_i\}\neq 0$ might lead to biased estimates of the fixed effects \citep{schumacher2020scale}.

\item[\bf iv)] The SMSN-NLME model defined in (\ref{modeleq11})-(\ref{modSnmis1}) can be formulated with a hierarchical representation, as follows:
\begin{eqnarray} \textbf{Y}_i|\be_i,U_i=u_i&\ind&
N_{n_i}(\eta(\A_i\bbeta+\be_i,\X_i),u_i^{-1}\sigma_e^2\mathbf{I}_{n_i}),\label{model0}\\
\be_i|U_i=u_i&\ind & SN_{q}(c\bDelta,
u_i^{-1}\bD,\blambda),\label{modelaa}\\
%T_i|U_i&\iid& TN_1(c,u_i^{-1})I(c,\infty)\label{model21},\\
U_i&\iid& H(\cdot;\bnu)\label{modelbb}.
\end{eqnarray}

% \item[\bf v)] As recommended by \cite{Lange89},
% \cite{Berkane94} and \cite{Lucas97},  who pointed out difficulties
% in estimating $\bnu$ due to problems of unbounded and local maximum
% in the likelihood function, in this work we take the value of
% $\bnu$ to be known.
\end{itemize}

Let
$\btheta=(\bbeta^{\top},\sigma_e^2,\balpha^{\top},\blambda^{\top},\bnu^\top)^{\top}$, then
classical inference on the parameter vector $\btheta$ is based on
the marginal distribution of
$\mathbf{Y}=(\Y^{\top}_1,\ldots,\Y^{\top}_n)^\top$
\citep{PinheiroBates95}. Thus, from the hierarchical representation in (\ref{model0})-(\ref{modelbb}), the integrated likelihood for $\btheta$ based on the observed sample $\yp = (\yp^{\top}_1,\ldots,\yp^{\top}_n)^\top$ in this case is given
by
\begin{eqnarray}
L(\btheta\mid \yp)&=&2 \,\prod^n_{i=1}\int^{\infty}_{0}\int_{\mathbb{R}^q}
\phi_{n_i}(\mathbf{y}_i;\eta(\A_i\bbeta+\be_i,\X_i),u_i^{-1}\sigma^2_e
\mathbf{I}_{n_i})\phi_q(\be_i;c\bDelta,u^{-1}_i\mathbf{D})\nonumber\\
&&\times \,\Phi(u_i^{1/2}\blambda^{\top}\bD^{-1/2}(\be_i-c\bDelta))d\be_idH(u_i;\bnu),\label{cor1eq}
\end{eqnarray}
which generally does not have a closed form expression because the
model function is nonlinear in the random effect. In the normal
case, in order to make the numerical optimization of the likelihood
function a tractable problem, different approximations to
(\ref{cor1eq}) have been proposed, usually based on first-order Taylor series expansion of the model function around the conditional mode of the random effects
\citep{Lindstrombates90}. Following this idea, we describe next 
two important results based on Taylor series approximation method
for approximating the likelihood function of a SMSN--NLME model. The
first uses a point in a neighborhood of $\ba_i$ as the expansion
point. The second uses simultaneously a neighborhood of $\ba$ and
$\bbeta$ as expansions points, with the advantage that this approximation is completely linear (in $\bbeta$ and $\ba$). These approximations can be considered as extensions of the result given in \cite{Lindstrombates90}, \cite{lin2017multivariate}, \cite{matos2013likelihood} and \cite{PinheiroBates95}.

\begin{theorem}\label{prop1} Let $\widetilde{\ba}_i$ be an
expansion point in a neighborhood of $\ba_i$, for $\ii$. Then,
under the SMSN--NLME model as given in (\ref{modeleq11})--(\ref{modSnmis1}),
the marginal distribution of $\mathbf{Y}_i$ can be approximated as follows:
\begin{equation}
\Y_i\app
\SNI_{n_i}\left(\eta(\A_i{\bbeta}+\widetilde{\ba}_i,\X_i)-\widetilde{\bH}_i(\widetilde{\ba}_i-c\bDelta),\widetilde{\bPsi}_i,\widetilde{\bar{\blambda}};H\right),
\end{equation}
where $\widetilde{\bPsi}_i=\widetilde{\bH}_i \bD
\widetilde{\bH}_i^{\top}+\sigma^2_e \mathbf{I}_{n_i}$,
$\widetilde{\bH}_i=\displaystyle\frac{\partial
\eta(\A_i{\bbeta}+{\ba}_i,\X_i)}{\partial{\ba}^{\top}_i}|_{\ba_i=\widetilde{\ba}_i},$
$\widetilde{\bar{\blambda}}_{i}=
\displaystyle\frac{\widetilde{\bPsi}_i^{-1/2}\mathbf{\widetilde{H}}_i\mathbf{D}{\bzeta}}
{\sqrt{1+\bzeta^{\top}\widetilde{\bLambda}_i\bzeta}},$
with $ \bzeta=\mathbf{D}^{-1/2}\blambda,
\widetilde{\bLambda}_i=\left(\mathbf{D}^{-1}+\sigma_e^{-2}\mathbf{\widetilde{H}}_i^{\top}
\mathbf{\widetilde{H}}_i\right)^{-1}$, and $`` \app"$ denotes approximated
in distribution.
\end{theorem}

\begin{proof}
For simplicity we omit the sub-index $i$. Thus, for $\bbeta$ fixed
and based on first-order Taylor expansion of the function $\eta$
around $\widetilde{\be}$, we have from (\ref{modeleq11}) that
$$\bepsilon=\Y-\eta(\A\bbeta+\be,\X)\approx \Y-\left[\eta(\A\bbeta+\widetilde{\be},\X)+\widetilde{\bH}(\be-\widetilde{\be})\right].$$
Then from (\ref{modeleq11}) and (\ref{modeleq2})
$$\Y-\left[\eta(\A\bbeta+\widetilde{\be},\X)+\widetilde{\bH}\be-\widetilde{\bH}\widetilde{\be}\right]~|~\be\app \NI_{n}(\mathbf{0},\sigma^2_e \mathbf{I} , H),$$
and the approximate conditional distribution of $\Y$ is
$$\Y~|~\be\app \NI_{n}(\eta(\A\bbeta+\widetilde{\be},\X)-\widetilde{\bH}\widetilde{\be}+\widetilde{\bH}{\be},\sigma^2_e \mathbf{I} , H),$$
or equivalently
$$\Y~|~\be,u\app \N_{n}(\eta(\A\bbeta+\widetilde{\be},\X)-\widetilde{\bH}\widetilde{\be}+\widetilde{\bH}{\be},u^{-1}\sigma^2_e \mathbf{I} ).$$
The rest of the proof follows by noting that
\begin{eqnarray*}
f(\y)&\approx& 2\int^{\infty}_{0}\int_{\mathbb{R}^q} \phi_{n}(\yp_i;\eta(\A\bbeta+\widetilde{\be},\X)-\widetilde{\bH}\widetilde{\be}+\widetilde{\bH}{\be},u^{-1}\sigma^2_e \mathbf{I} )\phi_q(\be;c\bDelta,u^{-1}\mathbf{D})\nonumber\\
&&\times\Phi(u^{1/2}\blambda^{\top}\bD^{-1/2}(\be-c\bDelta))d\be
dH(u;\bnu),
\end{eqnarray*}
which can be easily solved by using successively Lemmas 1 and 2
given in \cite{ArellanoLachos2005}.
\end{proof}

\begin{theorem}\label{prop2} 
Let $\widetilde{\ba}_i$ and $\widetilde{\bbeta}$ be expansion points in
a neighborhood of $\ba_i$ and $\bbeta$, respectively, for  $\ii,$. Then, under the
SMSN--NLME model as given in (\ref{modeleq11})--(\ref{modSnmis1}), the
marginal distribution of $\mathbf{Y}_i,$ can be
approximated as
\begin{equation}\label{margT2}
\Y_i\app
\SNI_{n_i}\left(\widetilde{\eta}(\widetilde{\bbeta},\widetilde{\ba}_i)+\widetilde{\bW}_i{\bbeta}+c\widetilde{\bH}_i\bDelta,\widetilde{\bPsi}_i,\widetilde{\bar{\blambda}};H\right),
\end{equation}
where
$\widetilde{\eta}(\widetilde{\bbeta},\widetilde{\ba}_i)=\eta(\A_i\widetilde{\bbeta}+\widetilde{\ba}_i,\X_i)-\widetilde{\bH}_i\widetilde{\be}_i-\widetilde{\bW}_i\widetilde{\bbeta}$,
\,$\widetilde{\bPsi}_i=\widetilde{\bH}_i \bD
\widetilde{\bH}_i^{\top}+\sigma^2_e \mathbf{I}_{n_i},$
\,$\widetilde{\bar{\blambda}}_{i}=\displaystyle\frac{\widetilde{\bPsi}_i^{-1/2}\mathbf{\widetilde{H}}_i\mathbf{D}\bzeta}
{\sqrt{1+\bzeta^{\top}\widetilde{\bLambda}_i\bzeta}},$
$\widetilde{\bH}_i=\displaystyle\frac{\partial
\eta(\A_i{\widetilde{\bbeta}}+{\ba}_i,\X_i)}{\partial{\ba}^{\top}_i}|_{\ba_i=\widetilde{\ba}_i},$
$\widetilde{\bW}_i=\displaystyle\frac{\partial
\eta(\A_i{{\bbeta}}+\widetilde{\ba}_i,\X_i)}{\partial{\bbeta}^{\top}}|_{\bbeta=\widetilde{\bbeta}},$
with $\bzeta=\mathbf{D}^{-1/2}\blambda,
\widetilde{\bLambda}_i=\left(\mathbf{D}^{-1}+\sigma_e^{-2}\mathbf{\widetilde{H}}_i^{\top}
\mathbf{\widetilde{H}}_i\right)^{-1}.$
\end{theorem}
\begin{proof}
As in Theorem \ref{prop1}, and based on first-order Taylor expansion of the
function $\eta$ around $\widetilde{\be}$  and $\widetilde{\bbeta}$,
we have that
$$\bepsilon=\Y-\eta(\A\bbeta+\be,\X)\approx \Y-\left[\eta(\A\widetilde{\bbeta}+\widetilde{\be},\X)+
\widetilde{\bH}(\be-\widetilde{\be})+
\widetilde{\bW}(\bbeta-\widetilde{\bbeta})\right].$$
Hence,
\begin{eqnarray}\label{linearizacao}
\Y~|~\be,U=u&\app&
\NI_{n}(\eta(\A\widetilde{\bbeta}+\widetilde{\be},\X)-\widetilde{\bH}\widetilde{\be}-\widetilde{\bW}\widetilde{\bbeta}+\widetilde{\bH}{\be}+\widetilde{\bW}{\bbeta},u^{-1}\sigma^2_e
\mathbf{I} , H),\nonumber\\ \be|U=u&\ind & SN_{q}(c\bDelta,
u^{-1}\bD,\blambda),\label{modelHT1}\\
%T_i|U_i&\iid& TN_1(c,u_i^{-1})I(c,\infty)\label{model21},\\
U&\iid& H(\cdot;\bnu)\nonumber,
\end{eqnarray}
and the proof follows by integrating out $(\be,u).$
\end{proof}

The estimates obtained by maximizing the approximate log-likelihood
function
$\ell(\btheta,\widetilde{\be})=\sumas\log{f(\yp_i;\btheta,\widetilde{\be}_i)}$
\,(or
$\ell(\btheta,\widetilde{\be},\widetilde{\bbeta})=\sumas\log{f(\yp_i;\btheta,\widetilde{\be}_i,\widetilde{\bbeta})}$)
 are thus approximate maximum likelihood estimates (MLEs), which can be computed directly through optimization procedures, such as {\it fmincon()} and {\it optim()} in {Matlab} and {R}, respectively. 
However, since numerical procedures for direct maximization of the approximate log-likelihood function
often present numerical instability and may not converge unless good starting values are used, in this paper we use the EM algorithm \citep{Dempster77} for obtaining approximate ML estimates via two modifications: the ECM algorithm \citep{Meng93} and the ECME algorithm \citep{Liu94}. 
%A key feature of the ECM algorithm is that it preserves the stability of the EM algorithm from its monotone convergence property and in addition it guarantees definite positive scale matrix estimate.

Before discussing the EM implementation to obtain ML estimates of a SMSN--NLME model, we present the empirical Bayesian estimate of the random effects $\widetilde{\be}^{(k)}$, which will be used in the estimation procedure and is given in the following result. The notation used is that of Theorem \ref{prop2} and the conditional expectations $\widetilde{\tau}_{-1i}$ can be easily derived from the result of Section 2 in \cite{Lachos_Ghosh_Arellano_2009}.

\begin{theorem} Let $\widetilde{\mathbf{Y}}_i=\mathbf{Y}_i-\widetilde{\eta}(\widetilde{\bbeta},\widetilde{\ba}_i)$, for $\ii$. Then the approximated minimum mean-squared error (MSE) estimator (or empirical Bayes estimator) of $\mathbf{b}_i$ obtained by the conditional mean of $\mathbf{b}_i$ given $\widetilde{\Y}_i=\widetilde{\y}_i$ is
\begin{eqnarray}\label{MSQ}
{\widehat{\mathbf{b}}}_i(\btheta)&\approx&\E\{\mathbf{b}_i|\widetilde{\Y}_i=\widetilde{\y}_i,\btheta\}=\widetilde{\bmu}_{bi}+\frac{\widetilde{\tau}_{-1i}}{\sqrt{1+\bzeta^{\top}\widetilde{\bLambda}_i\bzeta}}\,\widetilde{\bLambda}_i\bzeta,
\end{eqnarray}
where
$\widetilde{\bmu}_{bi}=c\bDelta+\mathbf{D}\widetilde{\bH}^{\top}_i\widetilde{\bPsi}^{-1/2}_i\widetilde{\y}_{0i}$
and
$\widetilde{\tau}_{-1i}=\E\left\{U^{-1/2}W_{\Phi}(U^{1/2}\widetilde{\A}_i)|\widetilde{\y}\right\}$,
with $W_{\Phi}(x)=\phi_1(x)/\Phi(x),\,x\in \mathbb{R}$,\,
$\widetilde{\y}_{0i}=\widetilde{\bPsi}^{-1/2}_i(\widetilde{\y}_i-\widetilde{\bW}_i\bbeta-c\widetilde{\bH}_i\bDelta)$
and
$\widetilde{\A}_i=\widetilde{\bar{\blambda}}^{\top}_{i}\widetilde{\y}_{0i}.$
\end{theorem}
\begin{proof}
From (\ref{modelHT1}), it can be shown that the conditional  distribution of the
$\mathbf{b}_i$ given $(\widetilde{\Y}_i,U_i)=(\widetilde{\y}_i,
u_i)$ belongs to the extended skew--normal (EST) family of
distributions \citep{Azzalini99}, and its pdf is
$$f(\mathbf{b}_i|\widetilde{\mathbf{y}}_i,u_i,\btheta)=\frac{1}{\Phi(u_i^{1/2}\widetilde{\A}_i)}\,
\phi_q(\mathbf{b}_i;\widetilde{\bmu}_{bi},u^{-1}_i\widetilde{\bLambda}_i){\Phi(u_i^{1/2}\bzeta^{\top}(\mathbf{b}_i-c\bDelta))}.$$
Thus, from Lemma 2 in \cite{Lachos_Ghosh_Arellano_2009}, we have
that
$$
\E\{\mathbf{b}_i|\widetilde{\y}_i,u_i,\btheta\}=\widetilde{\bmu}_{bi}+
\frac{u^{-1/2}_iW_{\Phi}(u_i^{1/2}\widetilde{\A}_i)}{\sqrt{1+\bzeta^{\top}\widetilde{\bLambda}_i\bzeta}}
\widetilde{\bLambda}_i\bzeta,
$$
and the MSE estimator of $\mathbf{b}_i$, given by
$\E\{\ba_i|\widetilde{\y}_i,\btheta\}$, follows by the law of
iterative expectations.
\end{proof}

% A disadvantage of direct maximization of the approximate log-likelihood function is that it
% may not converge unless good starting values are used. Thus, we use
% the EM algorithm \citep{Dempster77} for obtaining approximate ML estimates, which is
% less sensitive to stating values and is a powerful computational tool. The EM algorithm is stable and straightforward to implement since the iterations converge monotonically and no
% second derivatives are required. In this paper, we use the EM
% algorithm for parameter estimation via a simple modification, called
% ECM algorithm \citep{Meng93}. A key feature of the ECM algorithm is that it preserves the stability of the EM algorithm from its monotone convergence property and in addition it guarantees definite positive scale matrix estimate.
% Its implementation to obtain ML estimates of a SMSN--NLME model is discussed next.

\section{Approximates ML estimates via the EM algorithm}\label{sec3}

%In this section, we demonstrate how to use the EM--type algorithm to
%obtain approximate MLEs of a SNI--NLME model. 
Let the current estimate of $(\bbeta,\ba_i)$ be denote by
$(\widetilde{\bbeta},\widetilde{\ba}_i)$ and for simplicity hereafter we omit the symbol $``\sim"$ in $\bH_i$ and $\bW_i$. As in Theorem \ref{prop2}, the linearization procedure adopted in this section
consists of taking the first-order Taylor expansion of the nonlinear function
around the current parameter estimate $\widetilde{\bbeta}$ and
random effect estimate $\widetilde{\ba}_i$ at each iteration \citep{Wu2004,wu2009mixed}, which is equivalent to iteratively solving the LME model
\begin{equation}\label{aprox}
\widetilde{\mathbf{Y}}_i={\bW}_i\bbeta+{\bH}_i\ba_i+\bepsilon_i, \quad \ii,
\end{equation}
where $\widetilde{\mathbf{Y}}_i=\mathbf{Y}_i-\widetilde{\eta}(\widetilde{\bbeta},\widetilde{\ba}_i),$
$ \be_i\buildrel ind\over\sim \SNI_q(c\bDelta,\mathbf{D},\blambda,
H)$ {and} $\bepsilon_i\buildrel ind.\over\sim $ $
\NI_{n_i}(\mathbf{0},\sigma^2_e \mathbf{I}_{n_i} , H)$. 
A key feature of this model is that it can be formulated in a
flexible hierarchical representation that is useful for analytical
derivations. The model described in (\ref{aprox}) can be written as follows:
\begin{eqnarray} \widetilde{\textbf{Y}}_i|\textbf{b}_i,U_i=u_i\ind
N_{n_i}(\textbf{W}_i\bbeta+\textbf{H}_i\textbf{b}_i,u^{-1}_i\sigma_e^2\mathbf{I}_{n_i})&;&\,\,\textbf{b}_i|T_i=t_i,U_i=u_i\ind
N_{q}(\bDelta t_i,
u^{-1}_i\bGamma);\nonumber\\
T_i|U_i=u_i\ind TN(c,u_i^{-1};(c,\infty))&;&\,\,U_i\iid
H(\cdot;\bnu), \label{model21}
\end{eqnarray}
for $\ii,$ where $\bDelta=\textbf{D}^{1/2}\bdelta$,
$\bGamma=\textbf{D}-\bDelta\bDelta^{\top}$ with
$\bdelta=\blambda/(1+\blambda^{\top}\blambda)^{1/2}$ and
$\mathbf{D}^{1/2}$ being the square root of $\mathbf{D}$ containing
$q(q+1)/2$ distinct elements. $TN(\mu,\tau; (a,b))$ denotes the
univariate normal distribution $(N(\mu,\tau))$ truncated on the interval
$(a,b)$. 

Let $\mathbf{\widetilde{y}}_c=(\widetilde{\yp}^{\top},\mathbf{b}^{\top},\mathbf{u}^{\top},\mathbf{t}^{\top})^{\top}$,
with $\widetilde{\yp}=(\widetilde{\yp}^{\top}_1,\ldots,\widetilde{\yp}^{\top}_n)^{\top}$,
$\mathbf{b}=(\mathbf{b}^{\top}_1,\ldots,\mathbf{b}^{\top}_n)^{\top}$,
$\mathbf{u}=(u_1,\ldots,u_n)^{\top}$,
$\mathbf{t}=(t_1,\ldots,t_n)^{\top}$.
%and let
%$\btheta^{(k)}=(\bbeta^{{(k)}\top},\bgamma^{{(k)}\top},\balpha^{{(k)}\top},\blambda^{{(k)}\top})^{\top}$,
%denotes the estimates of $\btheta$ at the $k$--th iteration. 
%$\btheta=(\bbeta^{\top},\sigma_e^2,\balpha^{\top},\blambda^{\top})^{\top}$
It follows from (\ref{model21}) that the complete-data log-likelihood
function is of the form
\begin{eqnarray*}\label{logcompleta}
\ell_c(\btheta\mid\widetilde{\mathbf{y}}_c)&=&\sumas\left[-\frac{n_i}{2}\log{\sigma_e^2}-\frac{u_i}{2\sigma_e^2}(\widetilde{\yp}_i-\mathbf{W}_i\bbeta-\mathbf{H}_i\mathbf{b}_i)^{\top}(\widetilde{\yp}_i-\mathbf{W}_i\bbeta-\mathbf{H}_i\mathbf{b}_i)\right.\nonumber\\
&&\left.-\frac{1}{2}\log{|\bGamma|}-
\frac{u_i}{2}(\mathbf{b}_i-\bDelta
t_i)^{\top}\bGamma^{-1}(\mathbf{b}_i- \bDelta
t_i)\right]+K(\bnu)+C,
\end{eqnarray*}
where $C$ is a constant that is independent of the parameter vector $\btheta$ and $K(\bnu)$ is a function that depends on $\btheta$ only through $\bnu$.
Now, from (\ref{model21}) and by using successively Lemma 2
in \cite{ArellanoLachos2005} (see also \cite{Lachos_Ghosh_Arellano_2009}), it is straightforward to show that
\begin{eqnarray}
\mathbf{b}_i|t_i,u_i,\widetilde{\yp}_i,{\btheta}&\sim&
N_q(\mathbf{s}_it_i+\mathbf{r}_i,
u^{-1}_i\mathbf{B}_{i}),\nonumber\\
T_i|u_i,\widetilde{\yp}_i,\btheta&\sim&TN(c+\mu_{i},u^{-1}_i M^2_{i};(c,\infty)),\label{eqdens1}\\
\widetilde{\Y}_i|\btheta&\sim& \SNI_{n_i}(\bW_i\bbeta+c\,\bH_i\bDelta,\widetilde{\bPsi}_i,\widetilde{\bar{\blambda}};H),\nonumber
\end{eqnarray}
where
${M}_{i}=[1+{\bDelta}^{\top}\bH_i^{\top}{\bOmega}_i^{-1}\bH_i{\bDelta}]^{-1/2}$,
${\mu}_{i}={M}_{i}^2{\bDelta}^{\top}\bH_i^{\top}{\bOmega}_i^{-1}(\widetilde{\yp}_i-\mathbf{W}_i{\bbeta}-c\,\bH_i{\bDelta}
)$,
${\mathbf{B}}_{i}=[{\bGamma}^{-1}+\sigma_e^{-2}\displaystyle\bH_i^{\top}\bH_i]^{-1}$,
${\mathbf{s}}_i=(\mathbf{I}_q-\sigma_e^{-2}{\mathbf{B}}_{i}\bH_i^{\top}\bH_i){\bDelta}$,
${\bOmega}_i=\sigma_e^2\mathbf{I}_{n_i}+\bH_i{\bGamma}\bH^{\top}_i$,
and
$\mathbf{r}_i=\displaystyle\sigma^{-2}_e\mathbf{B}_{i}\bH_i^\top(\widetilde{\yp}_i-\bW_i\bbeta)$, for $\ii$. 

For the current value $\btheta = \widehat{\btheta}^{(k)}$, after some algebra the E-step of the EM algorithm can be written as
\begin{eqnarray*}
Q\left(\btheta\mid\widehat\btheta^{(k)}\right)&=&\E\left\{\ell_c(\btheta\mid\widetilde{\mathbf{y}}_c);\widehat{\btheta}^{(k)},\widetilde{\y}\right\} \\
&=&\sumas Q_{1i}\left(\btheta_1\mid\widehat{\btheta}^{(k)}\right)+\sumas
Q_{2i}\left(\btheta_2\mid\widehat{\btheta}^{(k)}\right) +\sumas
Q_{3i}\left(\bnu\mid\widehat{\btheta}^{(k)}\right),
\end{eqnarray*}
where $\btheta_1=(\bbeta^\top,\sigma^2_e)^\top$,
$\btheta_2=(\balpha^{\top},\blambda^{\top})^{\top}$,
\begin{eqnarray*}
Q_{1i}\left(\btheta_1\mid\widehat{\btheta}^{(k)}\right)&=&-\frac{n_i}{2}\log\widehat{\sigma}_e^{2(k)}
-\frac{1}{2\widehat{\sigma}_e^{2(k)}}\widehat{u}^{(k)}_i\left(\widetilde{\y}_{i}-\bW_i\widehat\bbeta^{(k)}\right)^{\top}
\left(\widetilde{\y_{i}}-\bW_i\widehat\bbeta^{(k)}\right)\nonumber\\
&&+\frac{1}{\widehat{\sigma}_e^{2(k)}}\left(\widetilde{\y}_{i}-\bW_i\widehat\bbeta^{(k)}\right)^{\top}\bH_i\widehat{(u\mathbf{b})}^{(k)}_i-\frac{1}{2\widehat{\sigma}_e^{2(k)}}\textrm{tr}\left\{\bH_i\widehat{(u\mathbf{b}\mathbf{b}^\top)}^{(k)}_i\bH^{\top}_i\right\},\\
Q_{2i}\left(\btheta_2\mid\widehat{\btheta}^{(k)}\right)&=&-\frac{1}{2}\log{|\widehat\bGamma^{(k)}|}
-\frac{1}{2}\textrm{tr}\left\{\widehat\bGamma^{-1(k)}\left(\widehat{(u\mathbf{b}\mathbf{b}^\top)}^{(k)}_i-
\widehat{(ut\mathbf{b})}^{{(k)}}_i\widehat\bDelta^{\top(k)}\right.\right.\\
&&\left.\left.-\widehat\bDelta^{{(k)}}
\widehat{(ut\mathbf{b})}^{\top(k)}_i+\widehat{(ut_2)}^{(k)}_i\widehat\bDelta^{(k)}\widehat\bDelta^{\top{(k)}}\right)\right\}\label{Q2mat},
\end{eqnarray*}
with $\textrm{tr}\{\mathbf{A}\}$ and $|\mathbf{A}|$ indicating the trace and determinant of matrix
$\mathbf{A}$, respectively. The calculation of these functions require expressions
for
$\widehat{u}^{(k)}_{i}=\E\{U_i|\widehat{\btheta}^{(k)},\widetilde{\yp}_i\}$,
$\widehat{(u\mathbf{b})}^{(k)}_i=\E\{U_i\mathbf{b}_i|\widehat{\btheta}^{(k)},\widehat{\yp}_i\}$,
$\widehat{(u\mathbf{b}\mathbf{b}^\top)}^{(k)}_i
=\E\{U_i\mathbf{b}_i\mathbf{b}^{\top}_i|\widehat{\btheta}^{(k)},\widetilde{\yp}_i\}$,
$\widehat{(ut)}^{(k)}_i=\E\{U_iT_i|\widehat{\btheta}^{(k)},\widetilde{\yp}_i\}$,
$\widehat{(ut_2)}^{(k)}_i=\E\{UT^2_i|\widehat{\btheta}^{(k)},\widetilde{\yp}_i\}$
and $\widehat{(ut\mathbf{b})}^{(k)}_i=\E\{U_iT_i
\mathbf{b}_i|\widehat{\btheta}^{(k)},\widetilde{\yp}_i\}.$  From
(\ref{eqdens1}), these can be readily evaluated as
\begin{eqnarray}
\widehat{(ut)}^{(k)}_i&=&
\widehat{u}_i^{(k)}(\widehat{\mu}^{(k)}_{i}+\widehat{c})+\widehat{M}^{(k)}_{i}\widehat{\tau}^{(k)}_{1i},\nonumber\\
\widehat{(ut_2)}^{(k)}_{i}&=&
\widehat{u}^{(k)}_i[\widehat{\mu}^{(k)}_{i}+\widehat{c}]^2+[\widehat{M}_{i}^{(k)}]^2+
\widehat{M}^{(k)}_{i}(\widehat{\mu}^{(k)}_{i}+2\widehat{c})\widehat{\tau}^{(k)}_{1i},\label{EMSNI}\\
\widehat{(u\mathbf{b})}^{(k)}_i&=
&\widehat{u}^{(k)}_i\widehat{\mathbf{r}}^{(k)}_i+\widehat{\mathbf{s}}^{(k)}_i
~\widehat{(ut)}^{(k)}_i,\,\,\,
\widehat{(ut\mathbf{b})}^{(k)}_i=\widehat{\mathbf{r}}^{(k)}_i~\widehat{(ut)}^{(k)}_i+
\widehat{\mathbf{s}}^{(k)}_i~\widehat{(ut_2)}^{(k)}_i,\nonumber\\
\widehat{(u\mathbf{b}\mathbf{b}^\top)}^{(k)}_i&=&
\widehat{\mathbf{B}}^{(k)}_{i}+\widehat{u}^{(k)}_i\widehat{\mathbf{r}}^{(k)}_i\widehat{\mathbf{r}}_i^{\top{(k)}}
+\widehat{\mathbf{r}}^{(k)}_i\widehat{\mathbf{s}}^{\top{(k)}}_i\widehat{(ut)}^{(k)}_i
+\widehat{\mathbf{s}}^{(k)}_i\widehat{\mathbf{r}}^{{(k)}
\top}_i\widehat{(ut)}_i+\widehat{\mathbf{s}}^{(k)}_i\widehat{\mathbf{s}}_i^{\top{(k)}}\widehat{(ut_2)}^{(k)}_i,\nonumber
\end{eqnarray}
where $\widehat{c} = c(\widehat{\bnu})$, and the expressions for $\widehat{ u}^{(k)}_i$ and
$\widehat{\tau}_{1i}=\E\{U_i^{1/2}W_{\Phi}({U_i^{1/2}\widehat{\mu}^{(k)}_{i}}/{\widehat{M}^{(k)}_{i}})| \widehat{\btheta}^{(k)},\widetilde{\yp}_i\}$ can be found in Section 2 from \cite{Lachos_Ghosh_Arellano_2009}, which can be easily implemented for the skew--$t$ and skew--contaminated normal distributions, but involve numerical integration for the skew--slash case.

The CM-step then conditionally maximize
$Q\left(\btheta\mid\widehat{\btheta}^{(k)}\right)$ with respect to $\btheta$,
obtaining a new estimate  $\widehat{\btheta}^{(k+1)}$, as follows:

\noindent{\bf CM-step 1:} Fix $\widehat{\sigma}_e^{2(k)}$ and update
$\widehat{\bbeta}^{(k)}$ as
\begin{eqnarray}
\widehat{\bbeta}^{(k+1)}&=&\left(\sumas\widehat{ u}^{(k)}_i
\bW_i^{\top}\bW_i\right)^{-1}\sumas
\bW_i^{\top}\left(\widehat{u}^{(k)}_i\widetilde{\yp}_i-\bH_i\widehat{(u\mathbf{b})}^{(k)}_i\right).
\label{maxmatlabmis00}
\end{eqnarray}

\noindent{\bf CM--step 2:} Fix  $\widehat{\bbeta}^{(k+1)}$ and
update $\widehat{\sigma}_e^{2(k)}$ as\begin{eqnarray*}
\widehat{\sigma}_e^{2{(k+1)}}&=&\frac{1}{N}\sumas
\left[\widehat{u}^{(k)}_i\left(\widetilde{\yp}_i-\mathbf{W}_i\widehat{\bbeta}^{(k+1)}\right)^{\top}
\left(\widetilde{\yp}_i-\mathbf{W}_i\widehat{\bbeta}^{(k+1)}\right)
\right.\\
&&\left.
-2\left(\widetilde{\yp}_i-\mathbf{W}_i\widehat{\bbeta}^{(k+1)}\right)^{\top}\bH_i\widehat{(u\mathbf{b})}^{(k)}_i
+\textrm{tr}\left\{\bH_i\widehat{(u\mathbf{b}\mathbf{b}^\top)}^{(k)}_i\bH^{\top}_i\right\}\right],
\end{eqnarray*} \vspace{-.5cm} 

\noindent where $N = \sumas n_i$.

\noindent{\bf CM--step 3:} Update $\widehat{\bDelta}^{(k)}$ as $
\widehat{\bDelta}^{(k+1)}=\displaystyle\frac{\sumas
\widehat{(ut\mathbf{b})}^{(k)}_i}{\sumas \widehat{(ut_2)}^{(k)}_i}.
$

\noindent{\bf CM--step 4:} Fix  $\widehat{\bDelta}^{(k+1)}$ and
update $\widehat{\bGamma}^{(k)}$ as {
\begin{equation*}
    \widehat{\bGamma}^{(k+1)}=\frac{1}{n}\sumas
\left(\widehat{(u\mathbf{b}\mathbf{b}^\top)}^{(k)}_i-\widehat{(ut\mathbf{b})}^{(k)}_i\widehat{\bDelta}^{\top{(k+1)}}-\widehat{\bDelta}^{(k+1)}\widehat{(ut\mathbf{b})}^{\top{(k)}}_i
+\widehat{(ut_2)}^{(k)}_i\widehat{\bDelta}^{(k+1)}\widehat{\bDelta}^{\top{(k+1)}}\right).
\end{equation*}
}

\noindent{\bf CM--step 5 (for ECME):}
Update $\widehat{\bnu}^{(k)}$ by
optimizing the constrained approximate log-likelihood
function (obtained from Theorem \ref{prop2}): $$
\widehat{\bnu}^{(k+1)}=\underaccent{\bnu}{\textrm{argmax}} \{\ell(\widehat{\btheta}^{*(k+1)},\bnu,\widetilde{\be}_i\mid \yp)\},$$
where $\btheta^{*} = \btheta \setminus \bnu$.  
%$\ell(\btheta,\widetilde{\be},\widetilde{\bbeta})$

It is worth noting that the proposed algorithm is computationally simple to
implement and it guarantees definite positive scale matrix
estimate, once at the $k$th iteration
$\widehat{\mathbf{D}}^{(k)}=\widehat{\bGamma}^{(k)}+\widehat{\bDelta}^{(k)}\widehat{\bDelta}^{\top{(k)}}$ and
$\widehat{\blambda}^{(k)}=\displaystyle \widehat{\mathbf{D}}^{-1/2(k)}\widehat{\bDelta}^{(k)}/
(1-\widehat{\bDelta}^{\top{(k)}}\widehat{\mathbf{D}}^{-1(k)}\widehat{\bDelta}^{(k)})^{1/2}$.
{The iterations are repeated until a suitable convergence rule is
satisfied, e.g., if
$||\widehat{\btheta}^{(k+1)}/\widehat{\btheta}^{(k)}-1||$ is
sufficiently small, or until some distance involving  two successive
evaluations of the approximate log-likelihood (derived from Theorem \ref{prop1}), like
$|\ell(\widehat{\btheta}^{(k+1)},\widetilde{\be}^{(k+1)})/\ell(\widehat{\btheta}^{(k)},\widetilde{\be}^{(k)})-1|$,
is small enough.}
Furthermore, $\widetilde{\yp}_i$, $\mathbf{W}_i$ and $\mathbf{H}_i$, for $\ii$, are updated in each step of the EM-type algorithm, with $\widetilde{\be}_i$ being computed at each iteration using \eqref{MSQ}.

{In addition, standard errors for $\widehat{\btheta}^{*}$ are estimated using the inverse of the observed information matrix obtained from the score vector following the results in \cite{schumacher2020scale} (see also \cite{skewlmm-manual}) and considering the linear approximation from Theorem \ref{prop2}. }

\subsection{Starting values}

It is well known that maximum likelihood estimation in nonlinear
mixed models may face some computational hurdles, in the sense that
the method may not give maximum global solutions if the starting
values are far from the real parameter values. Thus, the choice of
starting values for an EM-type algorithm in the nonlinear context plays
a big role in parameter estimation. In this work we consider the
following procedure for obtaining initial values for a SN--NLME model:

\begin{itemize}
\item Compute $\widehat{\bbeta}^{(0)}$ and $\widehat{\sigma}^{2(0)}_e$ and $\widetilde{\mathbf{b}}^{(0)}$ using the classical N--NLME model through the library {\it nlme()} in R software, for instance.
\item The initial value for the skewness parameter $\blambda$ is  obtained in the
following way: Let $\hat{\rho}_{l}$ be the sample skewness
coefficient of the $l$th column of $\widehat{\mathbf{b}}^{(0)}$, obtained
under normality. Then, we let $\widehat{\lambda}^{(0)}_{l}=3 \times
\textrm{sign}(\hat{\rho}_{l})$, $l=1,\ldots,q$.
\end{itemize}
Moreover, for ST--NLME, SCN--NLME or the SSL--NLME model we adopt the following strategy:
\begin{itemize}
\item Obtain initial values via method described above for the SN--NLME model;
\item Perform MLEs of the parameters of the SN--NLME via EM algorithm;
\item Use the EM estimates from the SN--NLME model as initial values for the corresponding ST--NLME, SSL--NLME and SCN--NLME models.
\item The initial values for ${\bnu}$ are considered as follows: $10$ for the ST distribution, $5$ for the SSL distribution, and $(0.05,0.8)$ for the SCN distribution.
% In order to estimate $\nu$ for the ST--NLME and SSL--NLME models, we computed the likelihood function for values of $\nu$ varying from 3 to 100 and 2 to 100 by 1,
% respectively, and chose the value of $\nu$ that maximizes the
% likelihood function. A similar procedure has been adopted for the
% SCN--NLM, with $\bnu$ varying in $(0,1)^2$.
\end{itemize}

Even though these procedures look reasonable for computing the
starting values, the tradition in practice is to try several
initial values for the EM algorithm, in order to get the highest
likelihood value. It is important to note that the highest
maximized likelihood is an essential information for some model
selection criteria, such as Akaike information criterion $(AIC,
-2\ell(\widehat{\btheta},\widetilde{\be}^{L})+2\aleph)$, where
$\aleph$ is the number of free parameters, which can be used in
practice to select between various SMSN--NLME models. In this work we
use the result from Theorem 1 to calculate AIC values.

\subsection{Futures observations\label{notes}}

Suppose now that we are interested in the prediction of
$\mathbf{Y}_i^+$, a $\upsilon\times 1$ vector of future measurements
of $\mathbf{Y}_i$, given the observed measurement
$\mathbf{Y}=(\mathbf{Y}^{\top}_{(i)},\mathbf{Y}^{\top}_i)^{\top}$,
where
$\mathbf{Y}_{(i)}=(\mathbf{Y}^{\top}_1,\ldots,\mathbf{Y}^{\top}_{i-1},\\\mathbf{Y}^{\top}_{i+1},\ldots,\mathbf{Y}^{\top}_{n})^\top$.
The minimum MSE predictor of $\Y^+_i$, which is the conditional
expectation $\Y_i^{+}$ given $\Y_i$ and $\btheta$, is given in the
following Theorem. The notation used is the one from Theorem \ref{prop1}.

\begin{theorem}\label{prop4} Let $\widetilde{\ba}_i$ be an expansion point in
a neighborhood of $\ba_i$, $\mathbf{Y}_i^+$ be an $\upsilon\times 1$
vector of future measurement of $\mathbf{Y}_i$  (or possibly
missing) and $\X^+_i$ be an $\upsilon\times r$ matrix of known
prediction regression variables. Then, under the SMSN--NLME model as
(\ref{modeleq11})--(\ref{modSnmis1}), the predictor (or minimum MSE
predictor) of $\mathbf{Y}_i^+$ can be approximated as
\begin{eqnarray}
\widehat{\mathbf{Y}}^+_i(\btheta)&=&\E\{\mathbf{Y}^+_i|\mathbf{Y}_i,\btheta\}\approx
\widetilde{\bmu}_{2.1}+\frac{\widetilde{\bPsi}_{i22.1}\bupsilon^{(2)}_i}{\sqrt{1+\bupsilon^{(2)\top}_i\widetilde{\bPsi}_{i22.1}\bupsilon^{(2)}_i}}\tau_{-1i},\label{condi}
\end{eqnarray}
where
$$\widetilde{\bmu}_{2.1}={\eta}(\A_i{\bbeta}+\widetilde{\ba}_i,{\X}^+_i)-\widetilde{\bH}^+_i({\widetilde{\ba}}_i-c\bDelta)+\widetilde{\bPsi}^*_{i21}\widetilde{\bPsi}^{*-1}_{i11}\left(\Y_i-{\eta}(\A_i{\bbeta}+\widetilde{\ba}_i,{\X}_i)+\widetilde{\bH}_i({\widetilde{\ba}}_i-c\bDelta)\right),$$
$\widetilde{\bPsi}_{i22.1}=\widetilde{\bPsi}^*_{i22}-\widetilde{\bPsi}^*_{i21}\widetilde{\bPsi}^{*-1}_{i11}\widetilde{\bPsi}^*_{i12},$
$\widetilde{\bPsi}^{*}_{i11}=\widetilde{\bPsi}_i=\widetilde{\bH}_i
\bD \widetilde{\bH}_i^{\top}+\sigma^2_e \mathbf{I}_{n_i}$,
$\widetilde{\bPsi}^*_{i12}=\widetilde{\bPsi}^{*\top}_{i21}=\widetilde{\bH}_i
\bD \widetilde{\bH}_i^{+\top}$,
$\widetilde{\bPsi}^{*}_{i22}=\widetilde{\bH}^+_i \bD
\widetilde{\bH}_i^{+\top}+\sigma^2_e \mathbf{I}_{\upsilon}$,\,
$\widetilde{\bPsi}^{*-1/2}_i\widetilde{\bar{\blambda}}_{i}^*=(\bupsilon^{(1)\top}_i,\bupsilon^{(2)\top}_i)^{\top}$, and
$$\tau_{-1i}=E\left\{U^{-1/2}_i\displaystyle{W_{\Phi}\left(U^{1/2}_i\widetilde{\bupsilon}^{\top}_i(\Y_i-{\eta}(\A_i{\bbeta}+\widetilde{\ba}_i,{\X}_i)+\widetilde{\bH}_i({\widetilde{\ba}}_i-c\bDelta))\right)}|\Y_i\right\},$$
with 
$\widetilde{\bPsi}^{*}_i = \left(\begin{array}{cc}
    \widetilde{\bPsi}^{*}_{i11} & \widetilde{\bPsi}^{*}_{i12} \\
     \widetilde{\bPsi}^{*}_{i21}&\widetilde{\bPsi}^{*}_{i22} 
\end{array}\right) = \sigma_e^2\mathbf{I}_{n_i+\upsilon}+\widetilde{\mathbf{H}}^*_i\mathbf{D}\widetilde{\mathbf{H}}_i^{*\top}$,\,
$\widetilde{\bar{\blambda}}^{*}_{i}=
\displaystyle\frac{\widetilde{\bPsi}_i^{*-1/2}\widetilde{\mathbf{H}}^*_i\mathbf{D}\bzeta}
{\sqrt{1+\bzeta^{\top}\widetilde{\bLambda}^*_i\bzeta}}$,\,
$\widetilde{\bLambda}^*_i=(\mathbf{D}^{-1}+\sigma_e^{-2}\widetilde{\mathbf{H}}^{*\top}_i
\widetilde{\mathbf{H}}^*_i)^{-1}$, 
$\widetilde{\bH}^*_i=(\widetilde{\bH}^{\top}_i,\widetilde{\bH}^{+\top}_i)^{\top}$,\,
$\widetilde{\bH}^+_i=\displaystyle\frac{\partial
\eta(\A_i{{\bbeta}}+{\ba}_i,\X^+_i)}{\partial{\ba}^{\top}_i}|_{\ba_i=\widetilde{\ba}_i},$
and\\\noindent
$\widetilde{\bupsilon}_i=\displaystyle\frac{{\bupsilon}^{(1)}_i+\widetilde{\bPsi}^{*-1}_{i11}\widetilde{\bPsi}^*_{i12}{\bupsilon}^{(2)}_i}{\sqrt{1+{\bupsilon}^{(2)\top}_i\widetilde{\bPsi}^{*}_{i22.1}{\bupsilon}^{(2)}_i}}.$
\end{theorem}

\begin{proof}
Under the notation and result given in Theorem \ref{prop1},  we have that
$$\Y_i^*=\left[\begin{array}{c}
    \mathbf{Y}_i \\
    \mathbf{Y}_i^+
  \end{array}\right]\app
  \SNI_{n_i+\upsilon}\left({\eta}(\widetilde{\ba}_i,\X^*_i)-\widetilde{\bH}^*_i({\widetilde{\ba}}_i-c\bDelta),\widetilde{\bPsi}^*_i,\widetilde{\bar{\blambda}}^*;H\right),
$$
where
${\eta}(\widetilde{\ba}_i,\X^*_i)=\left({\eta}^{\top}(\A_i{\bbeta}+\widetilde{\ba}_i,{\X}_i),{\eta}^{\top}(\A_i{\bbeta}+\widetilde{\ba}_i,{\X}^+_i)\right)^{\top}$,
${\X}^*_i=({\X}^{\top}_i,{\X}^{+\top}_i)^{\top}$. The  rest of the proof follows
by noting that $\Y_i^*|u_i\sim
SN_{n_i+\upsilon}({\eta}\left(\widetilde{\ba}_i,\X^*_i)-\widetilde{\bH}^*_i({\widetilde{\ba}}_i-c\bDelta),u_i^{-1}\widetilde{\bPsi}^*_i,\widetilde{\bar{\blambda}}^*\right)$
and applying the law of iterative expectations.
\end{proof}

It can be shown that marginally $\Y_i\app
SMSN\left({\eta}(\A_i{\bbeta}+\widetilde{\ba}_i,{\X}_i)-\widetilde{\bH}_i({\widetilde{\ba}}_i-c\bDelta),\widetilde{\bPsi}_i,\widetilde{\bPsi}^{-1/2}_i\widetilde{\bupsilon}_i;H\right)$, and thence
the conditional expectations $\tau_{-1i}$ can be easily derived from
the result of Section 2 from \cite{Lachos_Ghosh_Arellano_2009}.
In practice, the prediction of $\Y^+_i$ can be obtained by
substituting the ML estimate $\widehat{\btheta}$ and
$\widetilde{\be}^{L}_i$ into (\ref{condi}), that is 
$\widehat{\Y}^+_i=\widehat{\Y}^+_i(\widehat{\btheta},
\widetilde{\be}_i^{L})$, where $\widetilde{\be}_i^{L}$ is the random
effect estimate in the last iteration of the EM algorithm.\\

\section{Simulation studies}\label{sec4}

\begin{figure}[htb]
\begin{center}
%\centering \hspace{0.5cm}\centering \hspace{8cm} \\
\includegraphics[scale=.7]{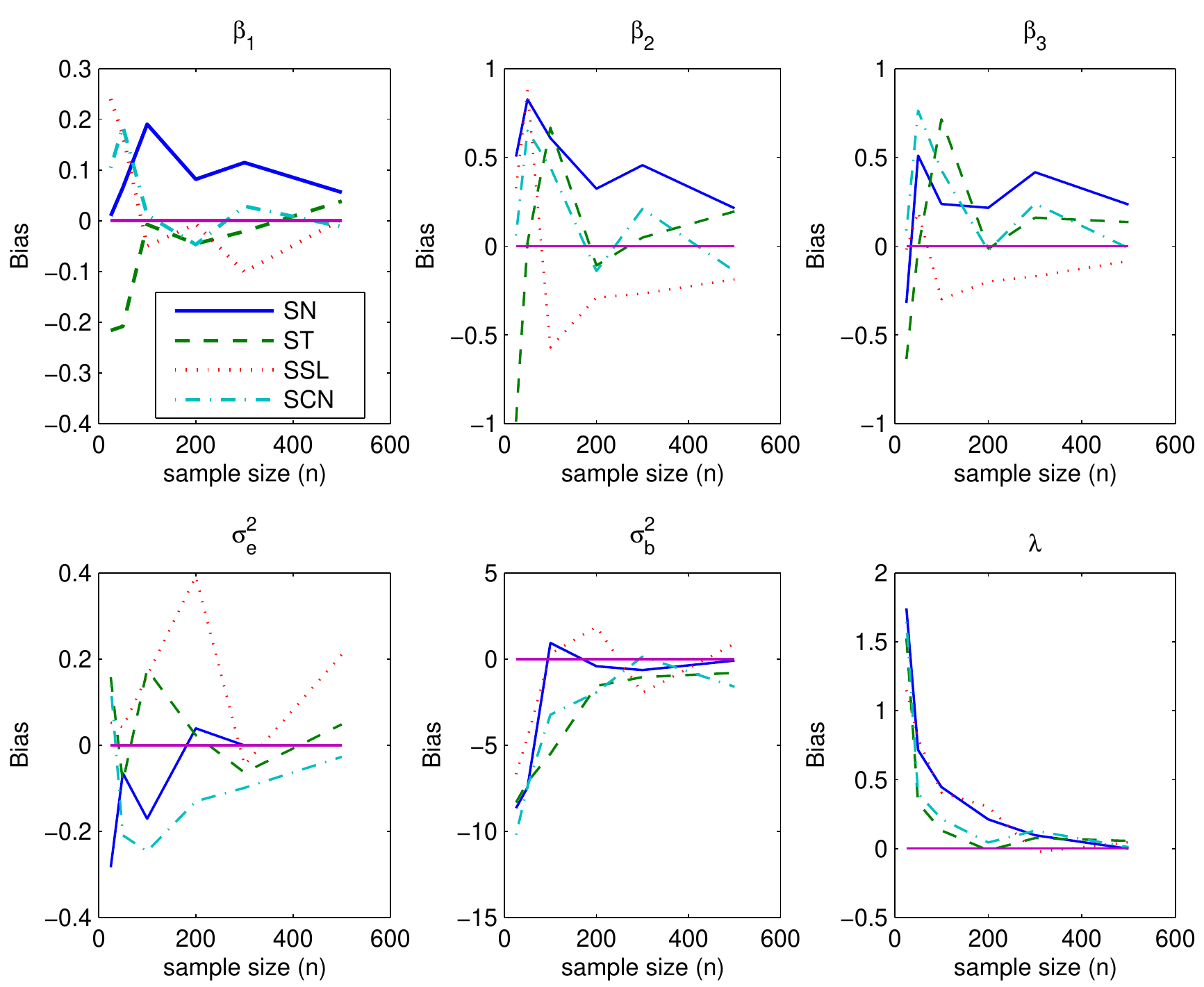}\\
\caption{Bias of the approximate ML estimates of $\bbeta,\sigma^2_e,\sigma_b^2$ and $\lambda$, based on $500$ Monte Carlo data sets for each SMSN distribution.\label{out1}}
\end{center}
\end{figure}

{ In order to examine the performance of the proposed method, in this section
we present the results of some simulation studies. 
{For simplicity, in the simulation studies we fix $\bnu$ at its true value.}
The first simulation study shows
that the proposed approximate ML estimates based on the EM algorithm
provide good asymptotic properties. The second study investigates
the consequences in population inferences of an inappropriate normality
assumption, and additionally it evaluates the efficacy of the measurement
used for model selection (AIC) when the result given in Theorem \ref{prop1} is used.}
\begin{figure}[htb]
\begin{center}
%\centering \hspace{0.5cm}\centering \hspace{8cm} \\
\includegraphics[scale=.7]{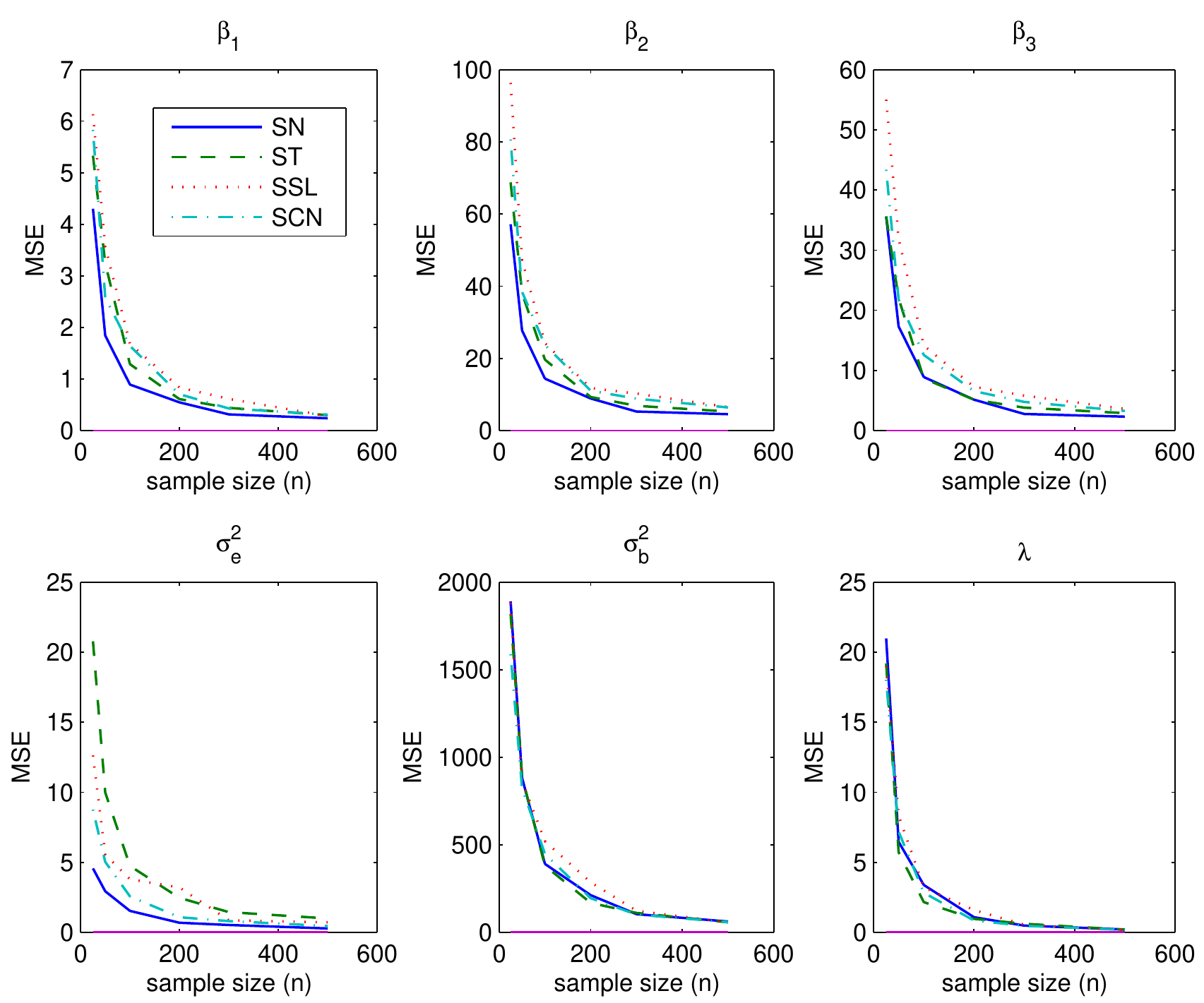}\\
\caption{MSE of the approximate ML estimates of $\bbeta,\sigma^2_e,\sigma_b^2$ and $\lambda$, based on $500$ Monte Carlo data sets for each SMSN distribution.\label{out2}}
\end{center}
\end{figure}

\subsection{First study}\label{sim1}

To evaluate the asymptotic behaviour of the proposed estimation method, we performed a simulation
study considering the following nonlinear growth-curve logistic model \citep{PinheiroBates95}:
\begin{equation}\label{simModel}
y_{ij}=\frac{\beta_1+b_{i}}{1+\exp{\{-(t_{j}-\beta_2)/\beta_3\}}}+\epsilon_{ij},\,\,
i=1,\ldots,n,\,\,\, j=1,\ldots,10,
\end{equation}
where $t_{j}=100, 267, 433, 600, 767, 933, 1100,1267,1433,1600$. The random effects $b_{i}$ and the error $\bepsilon_i=(\epsilon_{i1}\ldots,\epsilon_{i10})^{\top}$ are
non-correlated with
\begin{equation}\label{modSnmis1aux} \left(\begin{array}{c}
        b_i \\
        \bepsilon_i
      \end{array}
\right)\buildrel ind.\over\sim \SNI_{11}\left(\left(\begin{array}{c}
                               c\bDelta \\
                               \mathbf{0}
                             \end{array}
\right),\left(\begin{array}{cc}
           \sigma^2_b & \mathbf{0} \\
           \mathbf{0} & \sigma^2_e \mathbf{I}_{10}
         \end{array}
\right),\left(\begin{array}{c}
                               \lambda \\
                               \mathbf{0}
                             \end{array}
\right); H \right),\,\,i=1,\ldots,n.
\end{equation}
We set
$\bbeta=(\beta_1,\beta_2,\beta_3)^{\top}=(200,700,350)^{\top}$,
$\sigma_e^2=25$, $\sigma_b^2=100$, $\lambda=4$, implying in $\bDelta = 40/\sqrt{17} = 9.7014$, and $c=-\sqrt{{2}/{\pi}}\, k_1$, where $k_1$ depends on the specific SMSN distribution considered. Additionally, the samples sizes are fixed at $n=25,\,50,\,100,\,200,\, 300$ and $500$.
For each sample size, 500 Monte Carlo samples from the SMSN--NLME model in (\ref{modSnmis1aux}) are generated under four scenarios: under the skew--normal model (SN--NLME), under the skew--t with
$\nu=4$ (ST--NLME), under the skew--slash with $\nu=2$ (SSL--NLME), and under the
skew--contaminated normal model with $\bnu=(0.3,\,0.3)$ (SCN--NLME).
The values of $\bnu$ were chosen in order to yield a highly skewed
and heavy-tailed distribution for the random effects.

For each Monte Carlo sample, model (\ref{modSnmis1aux}) was
fit under the same distributional assumption that the data set was generated. 
Then we computed the empirical bias and empirical mean square error (MSE)
over all samples. For $\beta_1$, for instance, they are defined as
$$\text{Bias}(\bbeta_1)=\frac{1}{500}\sum^{500}_{k=1}\widehat{\beta}^{(k)}_1-\beta_1\,\,\textrm{and}\,\, \text{MSE}(\bbeta_1)=\frac{1}{500}\sum^{500}_{k=1}(\widehat{\beta}^{(k)}_1-\beta_1)^2, $$
respectively, where $\widehat{\beta}^{(k)}_1$ is the approximate ML estimate of $\beta_1$ obtained through ECM algorithm using the $k$th Monte Carlo sample. Definitions for the other parameters are obtained by analogy. 

Figures \ref{out1} and \ref{out2} show a graphical representation of the obtained results for bias and MSE, respectively. Regarding to the bias, we can see in general patterns of convergence to zero as $n$ increases. The worst case scenario seems to happen while estimating
the scale and skewness parameters of the random effect, which could be caused by the well known inferential problems related to the skewness parameter in skew--normal models, or maybe it would require a sample size greater than $500$ to obtain a reasonably pattern of convergence.
On the other hand, satisfactory values of MSE seem to occur when $n$ is greater than
400. As a general rule, we can say that both the bias and the MSE tend to approach to zero when the sample size is increasing, indicating that the approximate ML estimates based on the proposed EM-type algorithm provide good asymptotic properties.

\begin{figure}[htb]
\begin{center}
\centering
\includegraphics[width=0.6\textwidth]{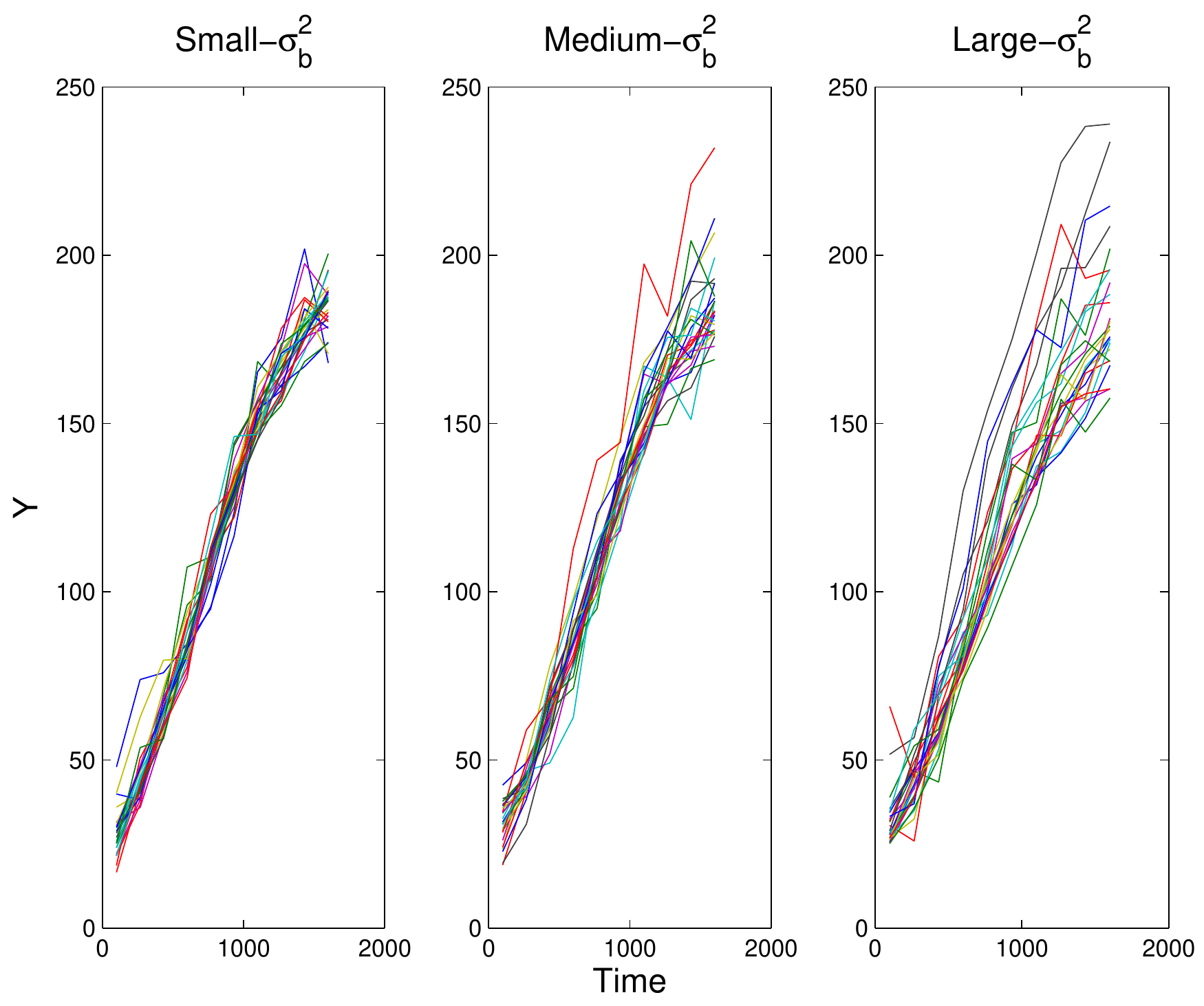}%,height=6cm
\caption{Simulated logistic curves under skew-t distribution for
different values of the scale parameter of the random effects.
\label{out1-2}}
\end{center}
\end{figure}

\subsection{Second study}%Comparison between N-NLME and ST-NLME models}

The goal of this simulation study is to asses the robustness or bias
incurred when one assumes a normal distribution for random effects
and the actual distribution is ST. The design of this simulation
study is similar to the one in Section \ref{sim1}, but now $500$ Monte Carlo samples
were generate considering only a ST model (\ref{modSnmis1aux}) with $\nu=4$
and $n=25$. Additional simulations were created by using the same
values of $(\bbeta,\sigma_e^2,\lambda)$ in (\ref{modSnmis1aux}) and
multiplying the scale parameter $\sigma^2_b$ by 0.25
and 6.25, obtaining $\sigma^2_b=25$ (small) and $\sigma_b^2=625$ (large).
This aims to verify if the proposed approximate methods are reliable in
different settings of the scale parameter $\sigma^2_b$. 
Therefore, three different scenarios are considered and for each scenario we fit model (\ref{simModel}) assuming the distributions normal and skew--$t$ with 4 degree of freedom, to each Monte Carlo data set.

For evaluating the capability of the proposed selection criteria in selecting the appropriate
distribution, the model preferred by the AIC criterion was also recorded for each sample.
Figure \ref{out1-2} shows example profiles for each of the three sizes
of scale components considered. The adjectives ``small'', ``medium'' and ``large'' are referring to the values assumed for $\sigma^2_b$. 
Note that for this particular model the variability increases with the mean as well as with the scale parameter.

\begin{table}[h] \small
\renewcommand{\baselinestretch}{1.5}
 \small
\begin{center}
\caption{\label{Norma_simula} Monte Carlo results for fixed effects
parameter estimates  based on $500$ Monte Carlo data generated from
a ST model (\ref{modSnmis1aux}) considering different values of the scale parameter $\sigma^2_b$ and $n=25$. True values of parameters are in parentheses and {\it pref. AIC} indicates the number of samples that each model was preferred by the AIC.}
\begin{tabular}{cccccccc}
 \hline
 &  &  \multicolumn{3}{c}{Normal model}&\multicolumn{3}{c}{ST model}  \\
 \cline{3-8}
Scenario      &  Measure   & $\beta_1$       & $\beta_2$     &  $\beta_3$ & $\beta1$ & $\beta_2$&$\beta_3$\\
     &     & (200)       & (700)     &  (350) & (200)& (700)&(350)\\
 \hline
  &Mean&      199.8313 & 698.6879 & 348.9097 & 199.8135 & 699.2931&  349.2883   \\
  &      Bias   &    -0.1687 &  -1.3121  & -1.0903&-0.1865   &-0.7069   &-0.7117   \\
Small--$\sigma_b^2$       &      MSE   & 6.0073 & 124.3313  & 72.4067& 3.1148 &  67.0638 &  39.8902 \\
       &     $95\%$ \textrm{Cov}   &    95.2   &95.0 & 95.2 &  96.4&  94.6  & 95.2\\
          &       {\it
pref. AIC}   &  \multicolumn{3}{c}{14}&\multicolumn{3}{c}{486} \\
%   &       {\it
%pref. BIC}   &  \multicolumn{3}{c}{25}&\multicolumn{3}{c}{475} \\
\hline
 &Mean&         199.6450 & 698.1040&  348.6829  & 199.7767 & 698.8985 & 349.0874  \\
  &      Bias   &  -0.3550   &-1.8960  & -1.3171&-0.2233  & -1.1015  & -0.9126   \\
Medium--$\sigma_b^2$        &      MSE   &8.6303 & 125.4999 &  71.8013 &   5.2957  & 73.9529   &41.6302 \\
       &     $95\%$ \textrm{Cov}   &     95.4   &94.4 &  94.6&  95.0&   95.6  & 94.2\\
          &       {\it
pref. AIC}   &  \multicolumn{3}{c}{11}&\multicolumn{3}{c}{489} \\
%   &       {\it
%pref. BIC}   &  \multicolumn{3}{c}{28}&\multicolumn{3}{c}{472} \\
\hline
  &Mean&          198.6449 & 696.4654 & 347.7706  & 199.3812 & 698.0883 & 348.7006   \\
  &      Bias   &   -1.3551  & -3.5346 &  -2.2294&-0.6188  & -1.9117 &  -1.2994  \\
 Large--$\sigma_b^2$&      MSE   &12.2231 & 121.7152  & 75.0350& 6.0414   &64.6235 &  36.6364
 \\
       &     $95\%$ \textrm{Cov}   &   90.4   &94.4 &  94.4 & 94.2&  94.8  & 94.4\\
          &       {\it
pref. AIC}   &  \multicolumn{3}{c}{26}&\multicolumn{3}{c}{474} \\
%   &       {\it
%pref. BIC}   &  \multicolumn{3}{c}{28}&\multicolumn{3}{c}{472} \\
\hline
\end{tabular}
\end{center}
\end{table}

Table \ref{Norma_simula} presents summary measures for the fixed effects
parameter estimates assuming normal and ST distributions for different values of the scale parameter $\sigma^2_b$, where the true parameters are indicated in parenthesis, Mean denotes the arithmetic average of the 500 estimates, Bias is the empirical mean bias, MSE is the empirical
mean squared error, and finally, $95\%$ Cov denotes the observed
coverage of the $95\%$ confidence interval computed using the
model-based standard error and the critical value=1.96.

The results in Table \ref{Norma_simula} suggest that irrespective of the fitted
NLME model, the bias and  MSE of the fixed effects increase as the
scale component becomes larger. Moreover, we notice from this table
that the bias and MSE from the ST fit are generally smaller than the ones from
the normal fit, indicating that models with skewness and
longer-than-normal tails may produce more accurate approximate MLEs.
In Figure \ref{out2-2} we present the empirical MSE for different
values of $n=25, 50, 100, 200$, and $500$ and for medium-$\sigma^2_b$, 
illustrating clearly the slower convergence to zero when the normal distribution is inappropriately used. 

Therefore, the results indicate that the efficiency in estimating fixed effects in NLME models can be severely degraded when normality is assumed, in comparison to considering a more flexible
approach via the ST distribution, corroborating with results from other authors, such as \cite{HartforDav200} and \cite{Litiere2007}. 
Since the main focus of such analysis is usually the evaluation of the fixed effects, this suggests that adopting normality assumptions routinely may lead to inefficient inferences on fixed effects when the true distribution is not normal. The inferences for the variance
components are not comparable for the two fitted models since they
are in different scales.

\begin{figure}[htb]
\begin{center}
%\centering \hspace{0.5cm}\centering \hspace{8cm} \\
\includegraphics[scale=.5]{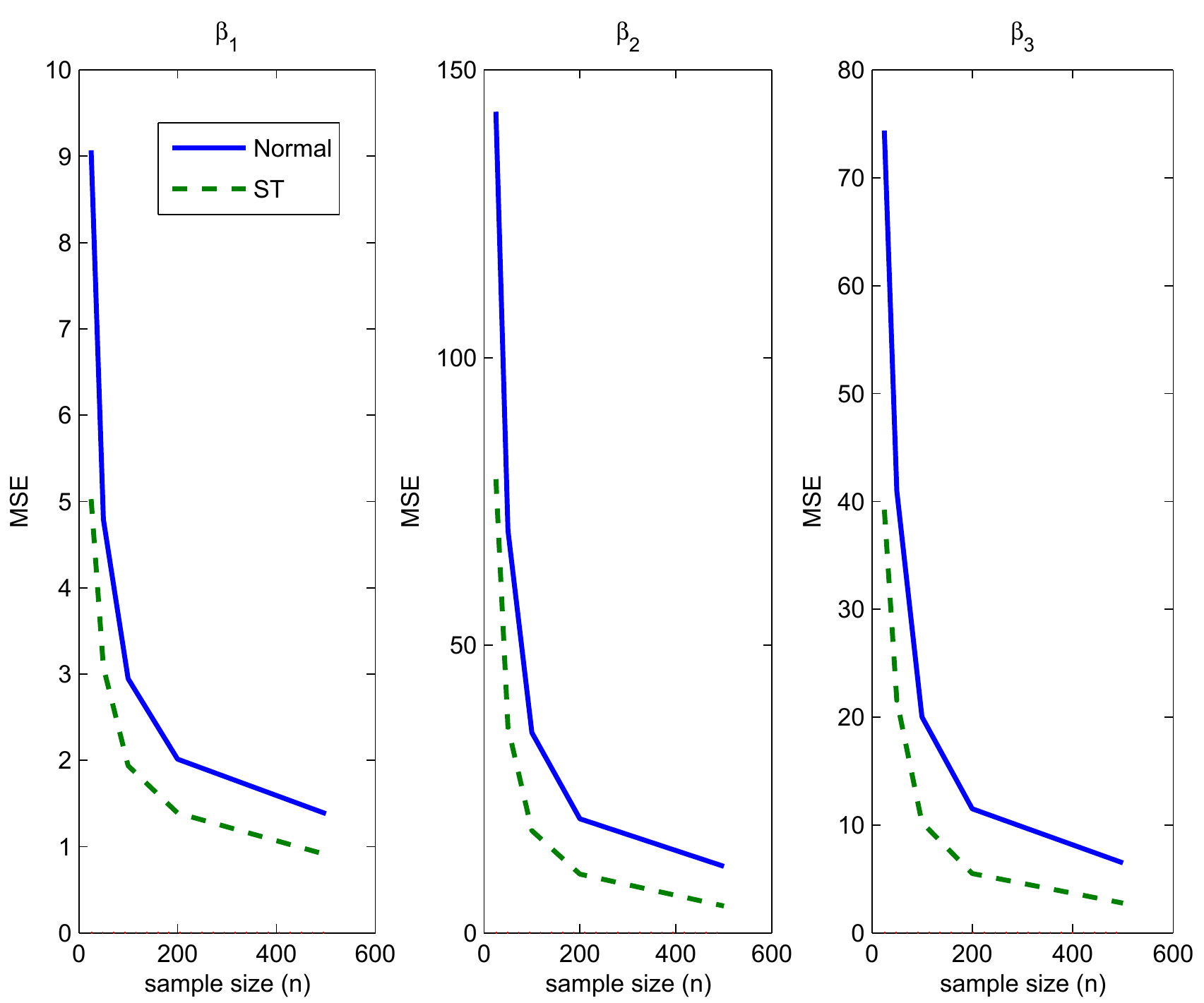}\\
\caption{MSE of the approximate ML estimates of $\beta_1,\beta_2$
and $\beta_3$, based on $500$ Monte Carlo data generated from a ST
model with $\sigma^2_b=100$ and for different
sample size $n$, when fitting a ST--NLME model (green line) and a N--NLME model (blue line).\label{out2-2}}
\end{center}
\end{figure}

\begin{figure}[htb]
\begin{center}
%\centering \hspace{0.5cm}\centering \hspace{8cm} \\
\includegraphics[scale=.4]{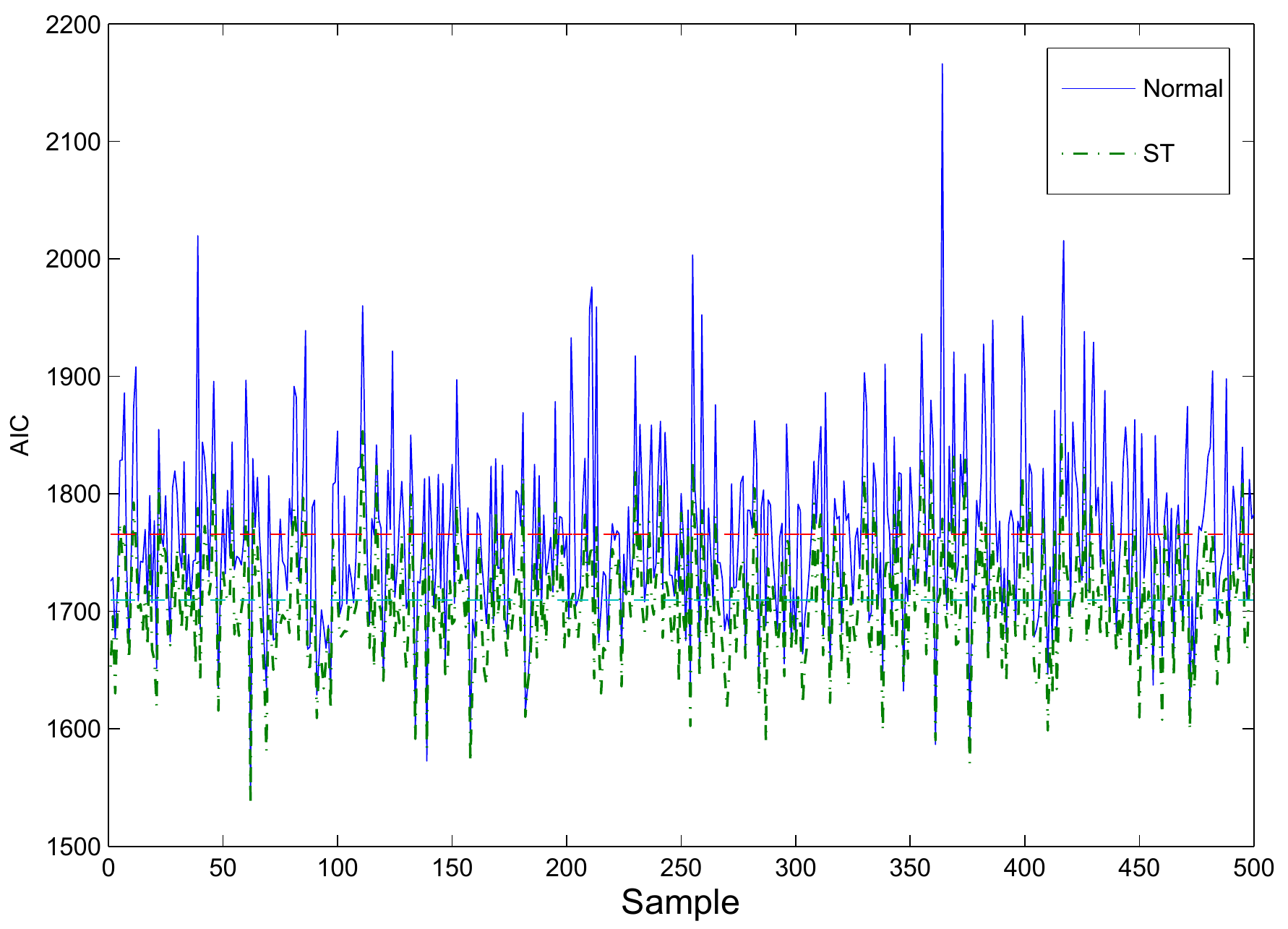}\\
\caption{AIC values for fitting a ST--NLME model (green line) and a N--NLME model (blue line), based on $500$ Monte Carlo data generated from a ST model with $\sigma^2_b=100$ and $n=25$. }\label{out22}
\end{center}
\end{figure}

Additionally, from Table \ref{Norma_simula} we can see that the AIC measure was able to classify the correct model well, indicating that the ST--NLME model presents a better fit than the N--NLME model, and the criteria for both models is illustrated in Figure \ref{out22}, where we show the AIC values for each sample and fitted model. 
Thence we conclude that the result given in Theorem \ref{prop1} provides a good approximation for the marginal likelihood function. In fact, this approximation is needed in order to make
the calculation of the AIC computationally feasible (and easy).\\

%\begin{figure}[!htb]
%\begin{center}
%\centering \hspace{0.5cm}\centering \hspace{8cm} \\
%\includegraphics[scale=.39]{pertbeta1.ps}~\includegraphics[scale=.4]{pertbeta2.ps}\\
%\includegraphics[scale=.37]{pertbeta3.ps}~\includegraphics[scale=.4]{sigmae.ps}\\
%\caption{{ {Simulated ST data. Relative changes on the ML estimates
%of $\bbeta$ and $\sigma^2_e$ from the N-NLME model (solid line) and
%the ST-NLME model (dashed-line) for different contaminations
%$\Im$.}}}\label{outpert}
%\end{center}
%\end{figure}

%aqui

\section{Theophylline kinetics data--{\it Theoph}}\label{sec5}

The Theophylline kinetics data set was first reported by \cite{boeckmann1994nonmem}, and it was previously analysed in \cite{PinheiroBates95} and \cite{PinheiroBates2000} by fitting a N--NLME model. In this section, we revisit the {\it Theoph} data with the aim of providing additional
inferences by considering SMSN distributions. In the experiment, the anti-asthmatic drug 
Theophylline was administered orally to 12 subjects whose serum
concentration were measured 11 times over the following 25 hours.
This is an example of a laboratory pharmacokinetic study
characterized by many observations on a moderate number of subjects.
Figure \ref{out3}(a) displays the profiles of the Theophylline
concentrations for the twelve patients.

% \begin{figure}[ht]
% \begin{center}
% \centering \hspace{0.6cm}(a)\centering \hspace{7cm} (b) \\
% \includegraphics[scale=.5]{Theoph.eps}~\includegraphics[scale=.5]{qqplotN.eps}\\
% \caption{{\it Theoph} data set. (a) Theophylline Concentration (in mg/L) versus time since oral administration of the drug in twelve patients. (b) Normal Q-Q plots of empirical Bayes
% estimates of $b_{1i}$ and $b_{2i}$.\label{out3}}
% \end{center}
% \end{figure}

\begin{figure}[ht]
\begin{center}
\includegraphics[width=.9\textwidth]{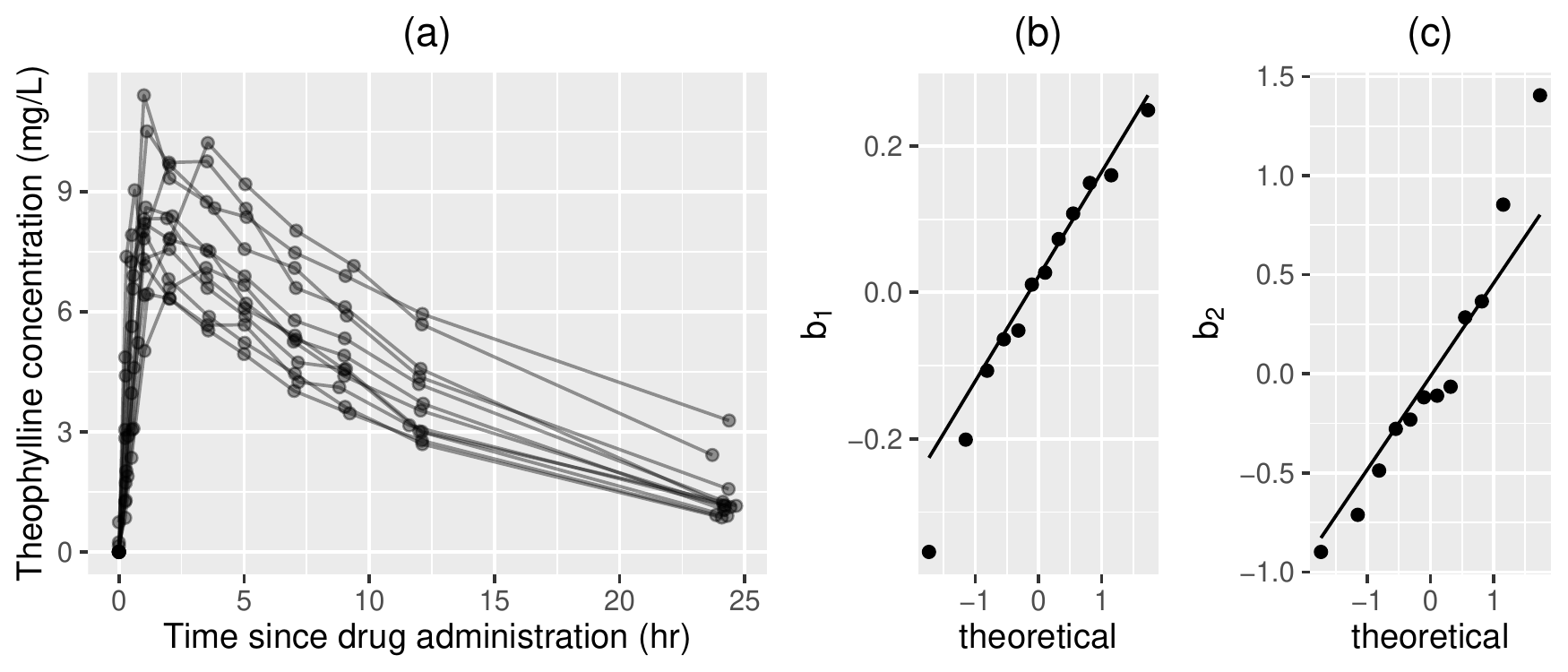}
\caption{{\it Theoph} data set. (a) Theophylline concentration (in mg/L) versus time since oral administration of the drug in twelve patients, and normal Q-Q plots of empirical Bayes
estimates of $b_{1i}$ (b) and $b_{2i}$ (c).\label{out3}}
\end{center}
\end{figure}

We fit a NLME model to the data considering the same nonlinear function as in \cite{PinheiroBates2000}, which can be written as
\begin{eqnarray}
C_{ij}&=&\frac{D_{i}\exp\{-(\beta_{1}+b_{i1})+(\beta_{2}+b_{i2})+\beta_{3}\}}{\exp(\beta_2+b_{i2})-\exp(\beta_3)}\hspace{5cm}\nonumber\\&&\times\left(\exp\{-\exp(\beta_{3})t_{ij}\}-\exp\{-\exp(\beta_2+b_{i2})t_{ij}\}\right) + \epsilon_{ij},
\end{eqnarray}
for $i=1,\ldots,12$, $j=1,\ldots,11$, where $C_{ij}$ represents the
$j$th observed concentration (mg/L) on the $i$th patient. $D_{i}$
represents the dose (mg/kg) administered orally to the $ith$ patient, and $t_{ij}$ is the time in hours. 
To verify the existence of skewness in the random effects, we start by fitting a traditional N--NLME model as in \cite{PinheiroBates2000}. Figures \ref{out3}(b) and \ref{out3}(c) depict
the Q-Q plots of the empirical Bayes estimates of $\mathbf{b}_i$ and shows that there
are some non-normal patterns on the random effects, including outliers and possibly skewness,  and therefore supporting the use of thick-tailed distributions.

Hence, we now consider a SMSN distribution for $\mathbf{b}_{i}$ and SMN distribution for ${\bepsilon}_{i}$, as in \eqref{modSnmis1}. Specifically, we consider the Normal, SN, ST, SCN and SSL distributions from the SMSN class for comparative purposes, and the results are presented next.
% For choosing the value of $\bnu$, we maximized the profile log-likelihood function, as illustrated in Figure \ref{fig10}. For the ST model we found $\nu=5$, for the SSL we found $\nu=1.3$ and for the SCN we found $\nu_1=0.3$ and $\nu_2=0.3$. 

% \begin{figure}[ht]
% \centering
% \includegraphics[scale=0.4]{nuST.eps}~\includegraphics[scale=0.4]{nuSSL.eps}\\
% \caption{{\it Theoph} data set. Plot of the profile log-likelihood for
% fitting a ST and SSL--NLME models.} \label{fig10}
% \end{figure}
%
Table \ref{ajuste1} contains the ML estimates of the parameters from the five models, together with their corresponding standard errors calculated via the observed
information matrix. The AIC measure indicates that heavy-tailed distributions present better fit that the Normal and SN--NLME models. Particularly, the model with ST distribution has the smaller AIC, being therefore the selected model. 
% Although the estimates of $\beta_1$, $\beta_2$ and $\beta_3$ are similar in all fitted
% models, the standard errors for the SMSN--NLME models with heavy tails are smaller than those in the normal and SN--NLME models. This suggests that the three models with longer tails than SN (and normal) seem to produce more accurate maximum likelihood estimates. 
{The standard errors of $\blambda$ are not reported since they are often not reliable (see \cite{schumacher2020scale}, for example),} and it is important to notice that the estimates for the variance
components are not comparable since they are on different scales.

\begin{table}[!htb]
 \centering
\small
\caption{ML estimation results for fitting various NLME models on
the {\it Theoph} data. SE denotes the estimated asymptotic standard errors
based on the observed information matrix. ($d_{11}, d_{12},
d_{22}$), are the distinct elements of the matrix
$\textbf{D}^{1/2}$. } \label{ajuste1}
\begin{tabular}{crrrrrrrrrr}\hline
&\multicolumn{2}{c}{N--NLME}&\multicolumn{2}{c}{SN--NLME}&\multicolumn{2}{c}{ST--NLME}&\multicolumn{2}{c}{SSL--NLME}&\multicolumn{2}{c}{SCN--NLME}\\\hline
Parameter&\multicolumn{1}{c}{MLE}&\multicolumn{1}{c}{SE}&\multicolumn{1}{c}{MLE}&\multicolumn{1}{c}{SE}&
\multicolumn{1}{c}{MLE}&\multicolumn{1}{c}{SE}&\multicolumn{1}{c}{MLE}&\multicolumn{1}{c}{SE}&\multicolumn{1}{c}{MLE}&\multicolumn{1}{c}{SE}\\
\hline
$\beta_1$&-3.228 & 0.066 & -3.232 & 0.239 & -3.200 & 0.163 & -3.214 & 0.180 & -3.195 & 0.137 \\
$\beta_2$&0.470 & 0.280 & 0.481 & 0.845 & 0.520 & 0.317 & 0.498 & 0.376 & 0.379 & 0.239 \\
$\beta_3$&-2.455 & 0.101 & -2.455 & 0.117 & -2.424 & 0.078 & -2.422 & 0.072 & -2.424 & 0.068 \\
$\sigma^2_e$&  0.503 & 0.049 & 0.502 & 0.057 & 0.297 & 0.114 & 0.165 & 0.059 & 0.208 & 0.056 \\
$d_{11}$& 0.167 & 0.072 & 0.212 & 0.222 & 0.226 & 0.192 & 0.164 & 0.140 & 0.182 & 0.126 \\
$d_{12}$&0.000 & 0.046 & -0.066 & 0.113 & -0.013 & 0.226 & -0.018 & 0.185 & 0.017 & 0.161 \\
$d_{22}$&0.644 & 0.239 & 0.784 & 0.447 & 0.714 & 0.440 & 0.525 & 0.280 & 0.522 & 0.290 \\
$\lambda_{1}$& &  & -2.740 &  & -28.605 &  & -27.143 &  & -26.482 &  \\
$\lambda_{2}$& &  & 2.677 &  & 7.997 &  & 9.415 &  & 3.152 &  \\
$\nu$ ($\nu_1$)& &  &  &  & 4.528 &  & 1.182 &  & 0.483 &  \\
$\nu_2$& && & &   &&  &&0.264&\\\hline
AIC&\multicolumn{2}{c}{368.044}&\multicolumn{2}{c}{369.676}&\multicolumn{2}{c}{{\bf 358.755}}&\multicolumn{2}{c}{360.657}&\multicolumn{2}{c}{359.748}\\
\hline
\end{tabular}
\end{table}

{ To asses the predictive performance of the N--NLME and SMSN--NLME
models, we remove sequentially the last few points of each response
vector, then we compute the ML estimates using the remaining data.
The deleted observations are considered as the true values to be
predicted.  As a measure of precision we use the mean of
absolute relative deviation $|(y_{ip}-\widehat{y}^+_{ip})/y_{ip}|$  (MARD),
where $p$ is the time point under forecast. 
%For instance, if we drop
% out the last eight measurement, then the prediction of
% $\yp_i=(y_{i4}, y_{i5}, y_{i6}, y_{i7}, y_{i8}, y_{i9}, y_{i10},
% y_{i11})^{\top},$ denoted by
% $\widehat{\mathbf{y}}^+_i=(\widehat{y}^+_{i4},\widehat{ y}^+_{i5},
% \widehat{y}^+_{i6}, \widehat{y}^+_{i7}, \widehat{y}^+_{i8},
% \widehat{y}^+_{i9}, \widehat{y}^+_{i10},
% \widehat{y}^+_{i11})^{\top},$ 
For instance, if we drop
out the last five measurements, then the prediction of
$\yp_i=(y_{i7}, y_{i8}, y_{i9}, y_{i10},y_{i11})^{\top},$ denoted by
$\widehat{\mathbf{y}}^+_i=(\widehat{y}^+_{i7}, \widehat{y}^+_{i8},
\widehat{y}^+_{i9}, \widehat{y}^+_{i10},
\widehat{y}^+_{i11})^{\top}$, is made using (\ref{condi}), for $i=1,\ldots,12$.
% The comparison of the predicted values to the real values based on the different
% models are given in Table \ref{ajuste3}. As expected, the results
% indicate that the SNI distribution with heavy tails yields better
% predictions than the SN and the Normal predictors. The ST model has
% a much smaller MARD than the normal and SN models for the
% Theophylline kinetics data; the relative improvement percentages are
% 18 and 15.9 per cent, respectively. 
Figure \ref{prevTheoph} presents the average of MARD in percentage $(\%)$ when the last $1,2,3,4$ and $5$ observations are deleted sequentially in each response vector and shows that the heavy-tailed SMSN models provide in general more accurate predictors than the normal model.
Particularly, when the last 5 observations are deleted for each subject, the difference between MARD from the ST and normal model is of almost $6\%$. 
Thus, the SMSN--NLME model with heavy-tailed distributions not only provides better
model fitting, it also yield smaller prediction errors for the
Theophylline kinetics data. 
%The SN-NLME model is noted to have the same
%behavior as those under the N-NLME model.
}

\begin{figure}[htb]
\begin{center}
\includegraphics[width=.5\textwidth]{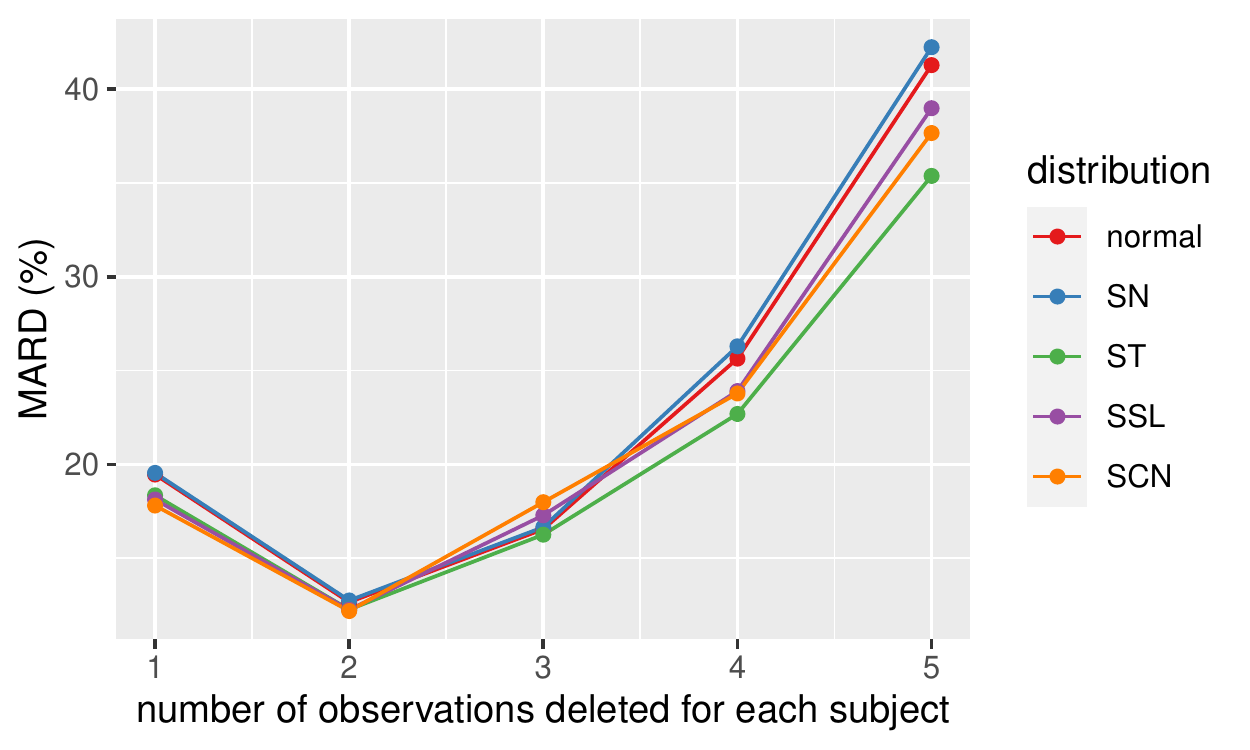}\\
\caption{{{ {\it Theoph} data set. Comparison of  forecast accuracy in terms
of MARD when the last $1,2,3,4$ and $5$ observations of each response
vector are deleted sequentially.}}}\label{prevTheoph}
\end{center}
\end{figure}

Furthermore, to assess the goodness of fit of the selected model, we construct a Healy-type plot \citep{healy1968multivariate}, by plotting the nominal probability values $1/n, 2/n, \hdots, n/n$ against the theoretical cumulative probabilities of the ordered observed Mahalanobis distances, which is calculated using the result Theorem \ref{prop1}. The Mahalanobis distances is a convenient measure for evaluating the distributional assumption of the response variable, once if the fitted model is appropriate the distribution of the Mahalanobis distance is known and given, for example, in \cite{schumacher2020scale}. If the fitted model is appropriate, the plot should resemble a straight line through the origin with unit slope. 
We also construct a Healy's plot for the Normal model for comparison, and the results are presented in 
Figure \ref{healyTheoph}.
It is clear that the observed Mahalanobis distances are closer to the expected ones in ST-NLME model than in the N-NMLE model, corroborating with the previous results.

\begin{figure}[htb]
\begin{center}
\includegraphics[width=.7\textwidth]{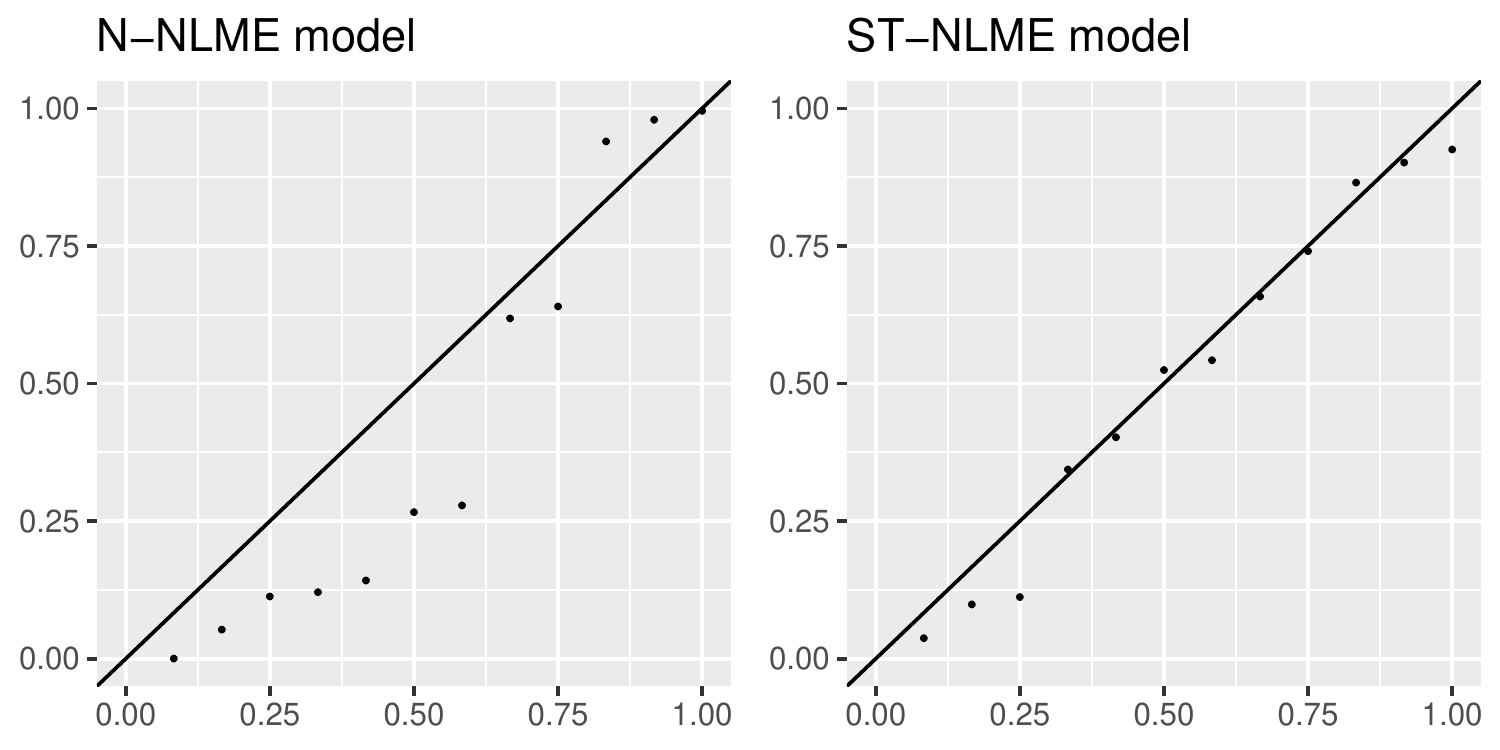}\\
\caption{{{ {\it Theoph} data set. Healy-type plots for assessing the goodness of fit of some SMSN-NLME models.}}}\label{healyTheoph}
\end{center}
\end{figure}

\section{Discussion and future works}\label{sec6}

Nonlinear mixed effects models are a research area with several challenging
aspects. In this paper, we proposed the application of a new
class of asymmetric distributions, called the SMSN class of distributions, to
NLME models. This enables the fit of a NLME model even when the
data distribution deviates from the traditional normal distribution.
Approximate closed-form expressions were obtained for the likelihood function of the observed data that can be maximized by using existing statistical software. An EM-type algorithm to obtain
approximate MLEs was presented, by exploring some important statistical
properties of the SMSN class. According to \cite{Wu2004}, in
complicated models, approximate methods are computationally more
efficient and may be preferable to the exact method, specially when
it exhibits convergence problems, such as slow convergence or
non-convergence. 

Furthermore, two simulation studies are presented, showing
the potential efficiency gain in fitting a more flexible model when the
normality assumption is violated. Moreover, in the analysis of the 
Theophylline data set the use of ST--NLME models offered better fitting as well as better prediction performance than the usual normal counterpart. Finally, we note that it may be worthwhile comparing our
results with other methods such as the classical Monte Carlo EM
algorithm or the stochastic version of the EM algorithm (SAEM), which
is beyond the scope of this paper. These issues will be considered
in a separate future work.  Another useful extension would be to consider a more general structure for the within-subject covariance matrix, such as an AR($p$) dependency structure as considered in \cite{schumacher2017censored}.

Finally, the method proposed in this paper is implemented in the software R \citep{rmanual}, and the codes are available for download from Github (\url{https://github.com/fernandalschumacher/skewnlmm}). 
We conjecture that the methodology presented in this paper should
yield satisfactory results in other areas where multivariate data
appears frequently, for instance: dynamic linear models, nonlinear
dynamic models, stochastic volatility models, etc., at the expense
of moderate complexity of implementation.

\section*{Acknowledgements}
Fernanda L. Schumacher acknowledges the partial support of Coordenação de Aperfeiçoamento de Pessoal de Nível Superior - Brasil (CAPES) - Finance Code 001, and by Conselho Nacional de Desenvolvimento Científico e Tecnológico - Brasil (CNPq).
%\end{acknowledgements}

\bibliographystyle{elsarticle-harv}%{unsrt}  

%\bibliography{bibli}   % name your BibTeX data base

\begin{thebibliography}{34}
\expandafter\ifx\csname natexlab\endcsname\relax\def\natexlab#1{#1}\fi
\providecommand{\url}[1]{\texttt{#1}}
\providecommand{\href}[2]{#2}
\providecommand{\path}[1]{#1}
\providecommand{\DOIprefix}{doi:}
\providecommand{\ArXivprefix}{arXiv:}
\providecommand{\URLprefix}{URL: }
\providecommand{\Pubmedprefix}{pmid:}
\providecommand{\doi}[1]{\href{http://dx.doi.org/#1}{\path{#1}}}
\providecommand{\Pubmed}[1]{\href{pmid:#1}{\path{#1}}}
\providecommand{\bibinfo}[2]{#2}
\ifx\xfnm\relax \def\xfnm[#1]{\unskip,\space#1}\fi
%Type = Article
\bibitem[{Arellano-Valle et~al.(2005)Arellano-Valle, Bolfarine and
  Lachos}]{ArellanoLachos2005}
\bibinfo{author}{Arellano-Valle, R.B.}, \bibinfo{author}{Bolfarine, H.},
  \bibinfo{author}{Lachos, V.}, \bibinfo{year}{2005}.
\newblock \bibinfo{title}{Skew-normal linear mixed models}.
\newblock \bibinfo{journal}{Journal of Data Science} \bibinfo{volume}{3},
  \bibinfo{pages}{415--438}.
%Type = Article
\bibitem[{Azzalini and Capitanio(1999)}]{Azzalini99}
\bibinfo{author}{Azzalini, A.}, \bibinfo{author}{Capitanio, A.},
  \bibinfo{year}{1999}.
\newblock \bibinfo{title}{Statistical applications of the multivariate
  skew-normal distribution}.
\newblock \bibinfo{journal}{Journal of the Royal Statistical Society}
  \bibinfo{volume}{61}, \bibinfo{pages}{579--602}.
%Type = Article
\bibitem[{Boeckmann et~al.(1994)Boeckmann, Sheiner and
  Beal}]{boeckmann1994nonmem}
\bibinfo{author}{Boeckmann, A.}, \bibinfo{author}{Sheiner, L.},
  \bibinfo{author}{Beal, S.}, \bibinfo{year}{1994}.
\newblock \bibinfo{title}{Nonmem users guide-part v: Introductory guide}.
\newblock \bibinfo{journal}{NONMEM Project Group. University of California at
  San Francisco} .
%Type = Article
\bibitem[{Branco and Dey(2001)}]{Branco_Dey01}
\bibinfo{author}{Branco, M.D.}, \bibinfo{author}{Dey, D.K.},
  \bibinfo{year}{2001}.
\newblock \bibinfo{title}{A general class of multivariate skew-elliptical
  distributions}.
\newblock \bibinfo{journal}{Journal of Multivariate Analysis}
  \bibinfo{volume}{79}, \bibinfo{pages}{99--113}.
%Type = Article
\bibitem[{De~la Cruz(2014)}]{delacruz2014bayesian}
\bibinfo{author}{De~la Cruz, R.}, \bibinfo{year}{2014}.
\newblock \bibinfo{title}{Bayesian analysis for nonlinear mixed-effects models
  under heavy-tailed distributions}.
\newblock \bibinfo{journal}{Pharmaceutical Statistics} \bibinfo{volume}{13},
  \bibinfo{pages}{81--93}.
%Type = Article
\bibitem[{Dempster et~al.(1977)Dempster, Laird and Rubin}]{Dempster77}
\bibinfo{author}{Dempster, A.}, \bibinfo{author}{Laird, N.},
  \bibinfo{author}{Rubin, D.}, \bibinfo{year}{1977}.
\newblock \bibinfo{title}{Maximum likelihood from incomplete data via the {EM}
  algorithm}.
\newblock \bibinfo{journal}{Journal of the Royal Statistical Society, Series
  B,} \bibinfo{volume}{39}, \bibinfo{pages}{1--38}.
%Type = Article
\bibitem[{Galarza et~al.(2020)Galarza, Castro, Louzada and
  Lachos}]{galarza2020quantile}
\bibinfo{author}{Galarza, C.E.}, \bibinfo{author}{Castro, L.M.},
  \bibinfo{author}{Louzada, F.}, \bibinfo{author}{Lachos, V.H.},
  \bibinfo{year}{2020}.
\newblock \bibinfo{title}{Quantile regression for nonlinear mixed effects
  models: a likelihood based perspective}.
\newblock \bibinfo{journal}{Statistical Papers} \bibinfo{volume}{61},
  \bibinfo{pages}{1281--1307}.
%Type = Article
\bibitem[{Hartford and Davidian(2000)}]{HartforDav200}
\bibinfo{author}{Hartford, A.}, \bibinfo{author}{Davidian, M.},
  \bibinfo{year}{2000}.
\newblock \bibinfo{title}{Consequences of misspecifying assumptions in
  nonlinear mixed effects models}.
\newblock \bibinfo{journal}{Computational Statistical \& Data Analysis}
  \bibinfo{volume}{34}, \bibinfo{pages}{139--164}.
%Type = Article
\bibitem[{Healy(1968)}]{healy1968multivariate}
\bibinfo{author}{Healy, M.}, \bibinfo{year}{1968}.
\newblock \bibinfo{title}{Multivariate normal plotting}.
\newblock \bibinfo{journal}{Journal of the Royal Statistical Society: Series C
  (Applied Statistics)} \bibinfo{volume}{17}, \bibinfo{pages}{157--161}.
%Type = Article
\bibitem[{Hui et~al.(2020)Hui, M{\"u}ller and Welsh}]{hui2020random}
\bibinfo{author}{Hui, F.K.}, \bibinfo{author}{M{\"u}ller, S.},
  \bibinfo{author}{Welsh, A.H.}, \bibinfo{year}{2020}.
\newblock \bibinfo{title}{Random effects misspecification can have severe
  consequences for random effects inference in linear mixed models}.
\newblock \bibinfo{journal}{International Statistical Review}
  \bibinfo{volume}{DOI: 10.1111/insr.12378}.
%Type = Article
\bibitem[{Lachos et~al.(2011)Lachos, Bandyopadhyay and Dey}]{lachos2011linear}
\bibinfo{author}{Lachos, V.H.}, \bibinfo{author}{Bandyopadhyay, D.},
  \bibinfo{author}{Dey, D.K.}, \bibinfo{year}{2011}.
\newblock \bibinfo{title}{Linear and nonlinear mixed-effects models for
  censored {HIV} viral loads using normal/independent distributions}.
\newblock \bibinfo{journal}{Biometrics} \bibinfo{volume}{67},
  \bibinfo{pages}{1594--1604}.
%Type = Article
\bibitem[{Lachos et~al.(2013)Lachos, Castro and Dey}]{lachos2013bayesian}
\bibinfo{author}{Lachos, V.H.}, \bibinfo{author}{Castro, L.M.},
  \bibinfo{author}{Dey, D.K.}, \bibinfo{year}{2013}.
\newblock \bibinfo{title}{Bayesian inference in nonlinear mixed-effects models
  using normal independent distributions}.
\newblock \bibinfo{journal}{Computational Statistics \& Data Analysis}
  \bibinfo{volume}{64}, \bibinfo{pages}{237--252}.
%Type = Article
\bibitem[{Lachos et~al.(2010)Lachos, Ghosh and
  Arellano-Valle}]{Lachos_Ghosh_Arellano_2009}
\bibinfo{author}{Lachos, V.H.}, \bibinfo{author}{Ghosh, P.},
  \bibinfo{author}{Arellano-Valle, R.B.}, \bibinfo{year}{2010}.
\newblock \bibinfo{title}{Likelihood based inference for skew--normal
  independent linear mixed models}.
\newblock \bibinfo{journal}{Statistica Sinica} \bibinfo{volume}{20},
  \bibinfo{pages}{303--322}.
%Type = Article
\bibitem[{Lin and Wang(2017)}]{lin2017multivariate}
\bibinfo{author}{Lin, T.I.}, \bibinfo{author}{Wang, W.L.},
  \bibinfo{year}{2017}.
\newblock \bibinfo{title}{Multivariate-nonlinear mixed models with application
  to censored multi-outcome aids studies}.
\newblock \bibinfo{journal}{Biostatistics} \bibinfo{volume}{18},
  \bibinfo{pages}{666--681}.
%Type = Article
\bibitem[{Lindstrom and Bates(1990)}]{Lindstrombates90}
\bibinfo{author}{Lindstrom, M.}, \bibinfo{author}{Bates, D.},
  \bibinfo{year}{1990}.
\newblock \bibinfo{title}{Nonlinear mixed-effects models for repeated-measures
  data}.
\newblock \bibinfo{journal}{Biometrics} \bibinfo{volume}{46},
  \bibinfo{pages}{673--687}.
%Type = Article
\bibitem[{Liti\`ere et~al.(2007)Liti\`ere, Alonso and
  Molenberghs}]{Litiere2007}
\bibinfo{author}{Liti\`ere, S.}, \bibinfo{author}{Alonso, A.},
  \bibinfo{author}{Molenberghs, G.}, \bibinfo{year}{2007}.
\newblock \bibinfo{title}{The impact of a misspecified random-effects
  distribution on the estimation and the performance of inferential procedures
  in generalized linear mixed models}.
\newblock \bibinfo{journal}{Statistics in Medicine} \bibinfo{volume}{27},
  \bibinfo{pages}{3125--31447}.
%Type = Article
\bibitem[{Liu and Rubin(1994)}]{Liu94}
\bibinfo{author}{Liu, C.}, \bibinfo{author}{Rubin, D.B.}, \bibinfo{year}{1994}.
\newblock \bibinfo{title}{The {ECME} algorithm: A simple extension of {EM} and
  {ECM} with faster monotone convergence}.
\newblock \bibinfo{journal}{Biometrika} \bibinfo{volume}{80},
  \bibinfo{pages}{267--278}.
%Type = Article
\bibitem[{Matos et~al.(2013)Matos, Prates, Chen and
  Lachos}]{matos2013likelihood}
\bibinfo{author}{Matos, L.A.}, \bibinfo{author}{Prates, M.O.},
  \bibinfo{author}{Chen, M.H.}, \bibinfo{author}{Lachos, V.H.},
  \bibinfo{year}{2013}.
\newblock \bibinfo{title}{Likelihood-based inference for mixed-effects models
  with censored response using the multivariate-t distribution}.
\newblock \bibinfo{journal}{Statistica Sinica} , \bibinfo{pages}{1323--1345}.
%Type = Article
\bibitem[{Meng and Rubin(1993)}]{Meng93}
\bibinfo{author}{Meng, X.}, \bibinfo{author}{Rubin, D.B.},
  \bibinfo{year}{1993}.
\newblock \bibinfo{title}{Maximum likelihood estimation via the {ECM}
  algorithm: A general framework}.
\newblock \bibinfo{journal}{Biometrika} \bibinfo{volume}{81},
  \bibinfo{pages}{633--648}.
%Type = Article
\bibitem[{Meza et~al.(2012)Meza, Osorio and De~la Cruz}]{meza2012estimation}
\bibinfo{author}{Meza, C.}, \bibinfo{author}{Osorio, F.},
  \bibinfo{author}{De~la Cruz, R.}, \bibinfo{year}{2012}.
\newblock \bibinfo{title}{Estimation in nonlinear mixed-effects models using
  heavy-tailed distributions}.
\newblock \bibinfo{journal}{Statistics and Computing} \bibinfo{volume}{22},
  \bibinfo{pages}{121--139}.
%Type = Article
\bibitem[{Pereira and Russo(2019)}]{pereira2019nonlinear}
\bibinfo{author}{Pereira, M.A.A.}, \bibinfo{author}{Russo, C.M.},
  \bibinfo{year}{2019}.
\newblock \bibinfo{title}{Nonlinear mixed-effects models with scale mixture of
  skew-normal distributions}.
\newblock \bibinfo{journal}{Journal of Applied Statistics}
  \bibinfo{volume}{46}, \bibinfo{pages}{1602--1620}.
%Type = Article
\bibitem[{Pinheiro and Bates(1995)}]{PinheiroBates95}
\bibinfo{author}{Pinheiro, J.}, \bibinfo{author}{Bates, D.},
  \bibinfo{year}{1995}.
\newblock \bibinfo{title}{Approximations to the log-likelihood function in the
  nonlinear mixed effects model}.
\newblock \bibinfo{journal}{Journal of Computational and Graphical Statistics}
  \bibinfo{volume}{4}, \bibinfo{pages}{12--35}.
%Type = Book
\bibitem[{Pinheiro and Bates(2000)}]{PinheiroBates2000}
\bibinfo{author}{Pinheiro, J.C.}, \bibinfo{author}{Bates, Douglas, M.},
  \bibinfo{year}{2000}.
\newblock \bibinfo{title}{Mixed-Effects Models in S and S-PLUS}.
\newblock \bibinfo{publisher}{Springer}, \bibinfo{address}{New York, NY}.
%Type = Article
\bibitem[{Pinheiro et~al.(2001)Pinheiro, Liu and Wu}]{Pinheiro01}
\bibinfo{author}{Pinheiro, J.C.}, \bibinfo{author}{Liu, C.H.},
  \bibinfo{author}{Wu, Y.N.}, \bibinfo{year}{2001}.
\newblock \bibinfo{title}{Efficient algorithms for robust estimation in linear
  mixed-effects models using a multivariate t-distribution}.
\newblock \bibinfo{journal}{Journal of Computational and Graphical Statistics}
  \bibinfo{volume}{10}, \bibinfo{pages}{249--276}.
%Type = Manual
\bibitem[{{R Core Team}(2020)}]{rmanual}
\bibinfo{author}{{R Core Team}}, \bibinfo{year}{2020}.
\newblock \bibinfo{title}{R: A Language and Environment for Statistical
  Computing}.
\newblock \bibinfo{organization}{R Foundation for Statistical Computing}.
  \bibinfo{address}{Vienna, Austria}.
\newblock \URLprefix \url{https://www.R-project.org/}.
%Type = Article
\bibitem[{Rosa et~al.(2003)Rosa, Padovani and Gianola}]{Rosa2003}
\bibinfo{author}{Rosa, G.J.M.}, \bibinfo{author}{Padovani, C.R.},
  \bibinfo{author}{Gianola, D.}, \bibinfo{year}{2003}.
\newblock \bibinfo{title}{Robust linear mixed models with normal/independent
  distributions and {B}ayesian {MCMC} implementation.}
\newblock \bibinfo{journal}{Biometrical Journal} \bibinfo{volume}{45},
  \bibinfo{pages}{573--590}.
%Type = Article
\bibitem[{Russo et~al.(2009)Russo, Paula and Aoki}]{Cibele09}
\bibinfo{author}{Russo, C.M.}, \bibinfo{author}{Paula, G.A.},
  \bibinfo{author}{Aoki, R.}, \bibinfo{year}{2009}.
\newblock \bibinfo{title}{Influence diagnostics in nonlinear mixed-effects
  elliptical models}.
\newblock \bibinfo{journal}{Computational Statistics and Data Analysis}
  \bibinfo{volume}{53}, \bibinfo{pages}{4143--4156}.
%Type = Article
\bibitem[{Schumacher et~al.(2017)Schumacher, Lachos and
  Dey}]{schumacher2017censored}
\bibinfo{author}{Schumacher, F.L.}, \bibinfo{author}{Lachos, V.H.},
  \bibinfo{author}{Dey, D.K.}, \bibinfo{year}{2017}.
\newblock \bibinfo{title}{Censored regression models with autoregressive
  errors: A likelihood-based perspective}.
\newblock \bibinfo{journal}{Canadian Journal of Statistics}
  \bibinfo{volume}{45}, \bibinfo{pages}{375--392}.
%Type = Article
\bibitem[{Schumacher et~al.(2020a)Schumacher, Matos and
  Lachos}]{schumacher2020scale}
\bibinfo{author}{Schumacher, F.L.}, \bibinfo{author}{Matos, L.A.},
  \bibinfo{author}{Lachos, V.H.}, \bibinfo{year}{2020}a.
\newblock \bibinfo{title}{Scale mixture of skew-normal linear mixed models with
  within-subject serial dependence}.
\newblock \bibinfo{journal}{arXiv preprint arXiv:2002.01040} .
%Type = Manual
\bibitem[{Schumacher et~al.(2020b)Schumacher, Matos and
  Lachos}]{skewlmm-manual}
\bibinfo{author}{Schumacher, F.L.}, \bibinfo{author}{Matos, L.A.},
  \bibinfo{author}{Lachos, V.H.}, \bibinfo{year}{2020}b.
\newblock \bibinfo{title}{skewlmm: Scale mixtures of skew-normal linear mixed
  models}.
\newblock \URLprefix \url{https://CRAN.R-project.org/package=skewlmm}.
  \bibinfo{note}{r package version 0.2.0}.
%Type = Article
\bibitem[{Verbeke and Lesaffre(1996)}]{VerbeLes1996}
\bibinfo{author}{Verbeke, G.}, \bibinfo{author}{Lesaffre, E.},
  \bibinfo{year}{1996}.
\newblock \bibinfo{title}{A linear mixed-effects model with heterogeneity in
  the random-effects population}.
\newblock \bibinfo{journal}{Journal of the American Statistical Association}
  \bibinfo{volume}{91}, \bibinfo{pages}{217--221}.
%Type = Article
\bibitem[{Wu(2004)}]{Wu2004}
\bibinfo{author}{Wu, L.}, \bibinfo{year}{2004}.
\newblock \bibinfo{title}{Exact and approximate inferences for nonlinear
  mixed-effects models with missing covariates}.
\newblock \bibinfo{journal}{Journal of the American Statistical Association}
  \bibinfo{volume}{99}, \bibinfo{pages}{700--709}.
%Type = Book
\bibitem[{Wu(2010)}]{wu2009mixed}
\bibinfo{author}{Wu, L.}, \bibinfo{year}{2010}.
\newblock \bibinfo{title}{Mixed Effects Models for Complex Data}.
\newblock \bibinfo{publisher}{Chapman and Hall/CRC, Boca Raton}.
%Type = Article
\bibitem[{Zhang and Davidian(2001)}]{ZhangDavi}
\bibinfo{author}{Zhang, D.}, \bibinfo{author}{Davidian, M.},
  \bibinfo{year}{2001}.
\newblock \bibinfo{title}{Linear mixed models with flexible distributions of
  random effects for longitudinal data}.
\newblock \bibinfo{journal}{Biometrics} \bibinfo{volume}{57},
  \bibinfo{pages}{795--802}.

\end{thebibliography}

\end{document}